\documentclass[lettersize,journal]{IEEEtran}
\IEEEoverridecommandlockouts
\ifCLASSINFOpdf
\else
\fi

\usepackage{tabularx}
\usepackage{threeparttable}
\usepackage{mathrsfs}
\usepackage{color}
\usepackage{amsfonts,amssymb}
\usepackage{amsfonts}
\usepackage{subfigure}
\usepackage{booktabs}

\usepackage{amsmath}
\usepackage{enumerate}
\usepackage{algorithm}
\usepackage{algorithmic}
\usepackage{bm}
\usepackage{makecell}
\usepackage{multirow}
\usepackage{mathtools}
\usepackage{bbm}
\usepackage{dsfont}
\usepackage{pifont}
\usepackage{amsthm}
\usepackage{cite}
\usepackage{balance}

\usepackage{siunitx}

\usepackage{tikz}
\usetikzlibrary{arrows.meta, positioning, shapes, calc}

\newtheorem{remark}{Remark}
\newtheorem{theorem}{Theorem}
\newtheorem{lemma}{Lemma}

\newtheorem{assumption}{Assumption}

\begin{document}

\title{Preserving Topology Privacy of Network Systems by Feedback: Conditions and Distributed Design}
\author{Yushan Li, Jiabao He, Julien M. Hendrickx, and Dimos V. Dimarogonas
\thanks{This work was supported by the Knut and Alice Wallenberg (KAW) Foundation, the Swedish Research Council (VR), and SIDDARTA concerted research action of the fédération Wallonie-Bruxelles.}
\thanks{Yushan Li, Jiabao He and Dimos V. Dimarogonas are with the Department of Decision and Control Systems, KTH Royal Institute of Technology, Stockholm, Sweden. Email: \{yushanl,jiabaoh,dimos\}@kth.se.} 
\thanks{Julien M. Hendrickx is with ICTEAM institute, UCLouvain, B-1348 Louvain-la-Neuve, Belgium. Email: julien.hendrickx@uclouvain.be.}
}

\maketitle

\begin{abstract}
This paper develops a feedback-based method to preserve the topology privacy of consensus protocols in network systems. 
The key idea is to intentionally violate topology identifiability conditions, thereby preventing unique or accurate recovery of the true topology from available observations, while preserving the intended consensus behavior. 
This problem is challenging because the feedback magnitude directly reflects the privacy level of edges, while it is strongly coupled with the consensus convergence and constrained by local communications at each node. 
To begin with, we derive the feedback conditions of both partial and full observation cases, 
where the topology unsolvability from observation data is characterized in the former, and the solution space that enforces topology inaccuracy from data is constructed in the latter. 
Then, we propose a novel distributed topology modification design under limited privacy budgets, and establish the performance guarantees through a controllable tradeoff between the consensus deviation and the topology privacy. 
Finally, we develop a low-complexity heuristic algorithm to achieve optimal privacy preservation on existing edges. 
Comparative simulations validate the effectiveness and outperformance of the proposed preservation design. 
\end{abstract}


\IEEEpeerreviewmaketitle

\section{Introduction}

In recent years, consensus protocols have become a core mechanism in numerous network systems (NSs) to fuse information, coordinate actions, and achieve global agreement, including sensor networks, robotic teams, and smart infrastructure \cite{olfati2007consensus}. 
In such consensus networks, the topology determines how state information is exchanged over the underlying communication links, and plays a central role in ensuring stability, robustness, and convergence. 
Hence, the topology structure is very sensitive information of NSs and its privacy needs to be well preserved.

\subsection{Motivations}
While the topology structure is fundamental to distributed coordination of consensus networks, 
the communication nature and the physical openness of these systems introduce a significant topology privacy vulnerability: an external eavesdropper or observer can collect the state trajectories by comprising the communication channels or measuring the operation situation of the physical plants. 
As demonstrated in network inference/identification studies \cite{brugere2018network,shvydun2024system}, these collected trajectories can enable external observers to recover the interaction topology, leading to severe privacy leakage on sensitive aspects of the communication infrastructure, such as network connectivity, routing structure, and interaction strengths \cite{hou2021combating,10210275}. 
For example, in wireless sensor network, one can use external radio-frequency sensor equipped with source separation and causal analysis to reconstruct directed communication links from overheard over-the-air traffic patterns \cite{9249412}. 
In robotic networks, an external observer can measure the motion trajectories of the robots to infer their interaction relationships \cite{sebastian2023learning}. 
To mitigate the privacy risk, this work aims to develop a topology privacy-preserving method that prevents external observers from recovering the true topology from available data, while preserving desired coordination performances.

\subsection{Related Works}

Topology inference, also referred to as network identification or reconstruction, has attracted extensive attention across various research communities. 
Existing approaches can broadly be grouped according to whether the underlying network dynamics are stochastic or deterministic. For stochastic network systems, a rich set of tools has been developed, including graph signal processing techniques \cite{mateos2019connecting,10146241}, causality-based \cite{9695344,li2024topology}, and vector autoregressive methods \cite{ioannidis2019semi,zaman2021online}. 
These methods typically rely on long horizon measurements and focus on the  asymptotic inference performance. 
In contrast, topology inference for deterministic dynamical networks often centers on understanding \emph{identifiability}, namely, determining under what conditions the network structure can be uniquely recovered from available data or excitation and measurement allocations \cite{hendrickx2019identifiability,cheng2022allocation,sun2024identifiability}. 
Representative efforts include node-perturbation or knock-out based strategies for undirected networks \cite{nabi-abdolyousefi2012networka}, as well as edge-level agreement conditions for nonlinear systems~\cite{10337619}. 
Identifiability criteria for general heterogeneous linear networks have been established in \cite{vanwaarde2021topology}. 

Compared with the extensive literature on topology inference, the problem of designing mechanisms that can resist such inference has received far less attention. 
Some straightforward ideas are to adopt random perturbations or encryption-based techniques in the consensus process as in \cite{mo2017privacy,nozari2017differentially,he2020differential,10261255}. 
These works achieved advanced progress on preserving the state privacy of the system \cite{wang2021dynamic,9999003,10360869}, but they rarely considered the topology level. 
The authors in \cite{katewa2020differential,Hawkins2024node} developed differential privacy based approaches to intentionally corrupt the broadcasted signals with noises, which were shown to protect the spectral or aggregate characteristics of the topology. 
The work \cite{li2025Preserving} further designed a noise-adding mechanism to preserve the topology matrix itself while maintaining state convergence in the mean-square sense. 
Nonetheless, these methods rely on persistent randomness injection, inherently creating a trade-off between the topology privacy preservation and convergence performance. 
In particular, if the continuous noise injection is used in the communication channel, it could increase signaling overhead and degrade link quality. 

Alternatively, the system identifiability analysis for deterministic networks can provide another perspective for the privacy preservation design by enforcing the topology unidentifiable, without relying on persistent randomization. 
For instance, \cite{mao2025unidentifiability} introduced a novel unidentifiability criterion based on the Fisher information matrix and proposed a low-rank controller to enforce this property in a single linear system, at the cost of nominal performance (e.g., slower or biased state convergence). 
The recent works \cite{hao2024discernibility,fan2024output} gave explicit conditions under which a change in the network topology can be distinguished from available measurement trajectories. 
However, how to use these conditions to intentionally preserve the topology privacy still remains an open issue. 

\subsection{Contributions}

In this work, we develop a feedback-based topology privacy preservation method, which prevents external observers from recovering the true topology from the available data, while preserving the intended consensus functionality. 
Specifically, the observer has no prior knowledge about the feedback design and the network structure. 
Building on the identifiability notion of dynamic networks in the state-space model \cite{vanwaarde2021topology}, 
the topology privacy is achieved by violating the identifiability conditions. 
Consequently, either multiple distinct topologies are consistent with the same data, or the topology identifiable from the data is biased away from the true one. 
The problem is challenging because 
i) the strong coupling between the topology and consensus convergence makes the feasible solution space of the introduced feedback difficult to characterize, 
and ii) the local communication constraints at each node further restrict the design. 
The key novelty lies in employing a locally implementable feedback modification without disrupting convergence or changing the communication pattern, while providing explicit privacy and performance guarantees. 

Preliminary results of analyzing the conditions of preserving the original consensus value and a Laplacian-based design have been presented in \cite{yushan2025Resistant}. 
This paper extends the idea to the general partial observation case of the observer, systematically characterizes the privacy conditions, and gives new distributed design under limited privacy budgets. 
The main contributions are summarized as follows. 
\begin{itemize}
\item  We identify the topology privacy vulnerability of 
consensus networks arising from their observed dynamic process. 
We consider both partial and full observation settings, which reflect different levels of accessible topology information.
Then, a feedback modification based preservation method is proposed, 
where the modification magnitude directly reflects the assigned privacy level and no new links are added. 

\item We derive conditions of the topology unsolvability from partial observations and the topology inaccuracy from full observations, which ensure the topology privacy while respecting local communication and convergence constraints. 
We characterize the resulting feasible solution space, and give sufficient explicit constructions of feedback. 
This feedback can be directly applied in centralized–design and distributed–execution architectures. 

\item To address the inherent tension between topology privacy and consensus convergence in a distributed manner, we introduce a controllable tradeoff that allows slight consensus deviations while guaranteeing the edge privacy. 
Then, we propose a distributed modification design under limited privacy budgets, where a low-complexity heuristic algorithm is developed to obtain the optimal privacy performance in hiding existing edges. 
Extensive simulations verify the effectiveness of the proposed method. 
\end{itemize}


Our method fundamentally differs from noise-adding methods \cite{katewa2020differential,Hawkins2024node,li2025Preserving}, which require persistent randomization and could still permit accurate topology identification from long-horizon observations. 
In contrast, our method requires only finite-time implementation and provides controllable guarantees on the privacy of edge weights. 
This method also conceptually echoes the topology discernibility \cite{hao2024discernibility,fan2024output}, but operates on complementary directions of intentional privacy design.

The remainder of this paper is organized as follows. 
In Section \ref{sec:Preliminaries}, some preliminaries and the problem of interest are presented. 
The topology privacy preservation conditions 
are given in Section \ref{sec:conditions}. 
Section \ref{sec:designs} presents the distributed topology privacy preservation design. 
Comparative simulations are shown in Section \ref{sec:simulation}.
Finally, Section \ref{sec:conclusion} concludes the paper.

\section{Preliminaries and Problem of Formulation}\label{sec:Preliminaries}

Consider a NS described by a digraph $\mathcal{G}=(\mathcal{V},\mathcal{E})$, where the vertex set $\mathcal{V}=\{1, \cdots, n\}$ enumerates $n$ interacting nodes, and the edge set $\mathcal{E} \subset \mathcal{V} \times \mathcal{V}$ characterizes directional edges. 
The presence of an edge $(i,j) \in \mathcal{E}$ signifies that node $i$ receives data from node $j$. 
The topology matrix of $\mathcal{G}$ is characterized by $W = [W_{ij}] \in \mathbb{R}^{n \times n}$, where $W_{ij} > 0$ denotes $(i,j) \in \mathcal{E}$, and $W_{ij} = 0$ otherwise. 
For node $i$, its incoming neighbors are defined as $\mathcal{N}_i = \{ j \in \mathcal{V} : (i,j) \in \mathcal{E} \}$, and $\mathcal{N}_i^+=\mathcal{N}_i \cup \{i\}$ if node $i$ is included. 

In this paper, we use $\otimes$ to denote the Kronecker product operator. 
Let $\bm{1}$ and $I$ be the all-one vector and identity matrix in compatible dimensions, respectively. 
The notation $\operatorname{rank}(\cdot)$ and $(\cdot)^\intercal$ denote the rank and transpose of a matrix, respectively. 
Given a series of vectors $v_1,\cdots,v_n$, the subspace spanned by these vectors is denoted by $\operatorname{span}\{v_1,\cdots,v_n\}$, 
and let $\operatorname{dim}(\cdot)$ be the dimension number of a vector space. 
For a matrix $A$, we use $A_{[i,:]}$ to denote the $i$-th row of $A$.

\subsection{Network System Model}
For distinction, we use the subscript $(\cdot)^*$ to denote the nominal NS without extra feedback, which is modeled as 
\begin{equation}\label{eq:local_model}
\begin{aligned}
x_{t+1}^{*,i} =  W_{ii} x_t^{*,i} + \sum\nolimits_{j\in \mathcal{N}_{i}} W_{ij} x_t^{*,j},
\end{aligned}
\end{equation}
where $x_{t}^{*,i}$ is the $i$-th node's state at time $t\in\{0,1,2,\cdots\}$. 
For the topology $W$, the following assumption is made. 
\begin{assumption}\label{assu:topo}
The topology matrix $W$ is row-stochastic, and the associated graph $\mathcal{G}$ is strongly connected. 
\end{assumption}

Based on Assumption \ref{assu:topo}, all nodes in \eqref{eq:local_model} will reach the following consensus point as $t\to\infty$ \cite[Theorem 5.1]{FB-LNS}
\begin{equation}\label{eq:consensus_p}
x^*_{\infty}=\lim_{t\to\infty} x_{t}^* = \lim_{t\to\infty} W^t x_0^* = (\pi^\intercal x_0^*) \bm{1},
\end{equation}
where $x_{t}^*=[x_{t}^{*,1},\cdots,x_{t}^{*,n}]^\intercal \in\mathbb{R}^n$ is the global system state, and $\pi \in\mathbb{R}^{n}$ is the normalized left eigenvector of the eigenvalue $1$ (also called the left-dominant eigenvector), satisfying $\pi^\intercal \bm{1}=1$. 
The network model \eqref{eq:local_model} is popular in many areas, e.g., formation forming and tracking control of mobile robots, opinion dynamics in social networks, and frequency regulation in power grids, etc \cite{olfati2007consensus}.


\subsection{Topology Identifiability from Observed Data}\label{subsec:infer_model}

Suppose that an external observer collects consecutive $T$ steps of states
 of the nominal NS \eqref{eq:local_model}, and the observations are given by  
\begin{equation}\label{eq:observation}
y_t^*=C x_t^*,~t=0,\cdots,T-1,
\end{equation}
where $C\in\mathbb{R}^{m\times n}$ ($m\le n$) is the observation matrix determined by the capability of the external observer. 
Specifically, we assume that $(A,C)$ is observable, which is standard in system identification literature \cite{ljung1999system}, but more general than those assuming that all states of the NS are measurable in many topology inference works \cite{mateos2019connecting,mo2017privacy,he2020differential}. 
Considering this partial observation capability is meaningful, 
because an external observer may not collect all state information, e.g., it cannot overhear every state broadcast due to limited radio-frequency coverage in wireless sensor networks. 
Then, the following observability matrix $Q_o^*$ has full column rank,
\begin{equation}\label{eq:Qo}
Q_o^*=\left[C^\intercal ,(CW)^\intercal ,\cdots, (C W^{n-1})^\intercal  \right]^\intercal  \in \mathbb{R}^{nm \times n}.
\end{equation}
By stacking $n$ steps of future observations, we have
\begin{equation}\label{eq:output-state}
\tilde{y}_{t:t+n-1}^*\!=\! [(y_t^*)^\intercal , (y_{t+1}^*)^\intercal , \cdots, (y_{t+n-1}^*)^\intercal ]^\intercal \! =\! Q_o^* x_t^*.
\end{equation}
Furthermore, it is convenient to use a block Hankel matrix, that collects all $T \ge nm + n - 1$ steps of observations into
\begin{equation} \label{eq:Hankel}
    \tilde{Y}_{0:T-n}^{*}=[\tilde{y}_{0:n-1}^{*},\tilde{y}_{1:n}^{*},\cdots,\tilde{y}_{T-n:T-1}^{*}] \in \mathbb{R}^{n m \times  (T-n+1)}.
\end{equation}
In the system identification literature, the data matrix $\tilde{Y}_{0:T-n}^{*}$ is generally required to have rank $n$, such that $W$ is identifiable up to a similarity transformation \cite{vanwaarde2021topology,shvydun2024system}, given by
\begin{align}\label{eq:estimator1}
    Q_W=g(\tilde{Y}_{0:T-n}^{*}),
\end{align}
where $g(\cdot)$ is a composite function of $\tilde{Y}_{0:T-n}^{*}$. 
The details of deriving $Q_W$ are provided in Appendix \ref{appen:inference}. 
Note that the identifiability from data is also referred to as data informativity~\cite{van2023informativity}, and for convenience, no distinction is made between the two in this paper. 
We remark that although $Q_W$ does not necessarily equal to $W$, it embeds sensitive spectral information about the topology. 
Therefore, preserving privacy of $Q_W$ constitutes an indirect but highly practical form of topology privacy.

Particularly, if the observer has full-state access ($C=I$), 
$W$ is uniquely identifiable when the collected data matrix $X_{0:T-1}^*=[x_{0}^*,x_{1}^*,\cdots,x_{T-1}^*]\in \mathbb{R}^{n \times T}$ has full rank $n$,
which can be represented by 
\begin{equation}\label{eq:old_estimator}
W = X_{1:T}^* (X_{0:T-1}^*)^{\dag},
\end{equation} 
where $X_{1:T}^*=[x_{1}^*,x_{2}^*,\cdots,x_{T}^*]$ and $(\cdot)^\dag$ represents the pseudo-inverse of a matrix. 
In this case, the indirect privacy preservation of $Q_W$ reduces to the direct preservation of $W$. 
Since this case corresponds to the strongest observer, it is also of great importance to deal with this worst-case privacy breach.

\subsection{Problem of Interest}

Based on the above topology identifiability from data, 
we aim to preserve topology privacy by intentionally violating the identifiability conditions, such that the available observations to external observers are insufficient to uniquely recover the true topology. 
Specifically, our key idea is to introduce a feedback design in the nodes' dynamic process, described by
\begin{equation}\label{eq:local_feedback}
\left\{\begin{aligned}
u^{i}_t &=\sum\nolimits_{j\in\mathcal{N}_i} K_{ij} x^{j}_t \\
x_{t+1}^{i} &=  W_{ii} x_t^{i} + \sum\nolimits_{j\in \mathcal{N}_{i}} W_{ij} x_t^{j} + u^{i}_t
\end{aligned}\right.,
\end{equation}
where $x_0=x_0^*$ and $K=[K_{ij}]_{i,j=1}^{n}\in\mathbb{R}^{n\times n}$ is the nonzero feedback matrix supported on communication links.  
The introduced $u_i$ serves as a local perturbation on a node to obfuscate external observers, e.g., perturbing a social user's opinion update in opinion dynamics \cite{abawajy2016privacy} or the velocity update of robots in flocking tasks \cite{Wang2024Topology}. 
Note that the feedback $K$ needs to be well designed to preserve the original consensus point \eqref{eq:consensus_p}, which is crucial for the reliability of consensus networks. 
Accordingly, the global form of \eqref{eq:local_feedback} is written as
\begin{equation}\label{eq:global_feedback}
x_{t+1}= (W+K) x_t \triangleq \tilde{W} x_t. 
\end{equation}
For the external observer, the superscript $(\cdot)^*$ on notations in Sec. \ref{subsec:infer_model} (e.g., $y_t^*$, $\tilde{Y}_{0:T-n}^*$, $X_{0:T-1}^*$) are dropped to match the dynamics \eqref{eq:global_feedback}.

For the case of an observable NS, we preserve the indirect-level privacy of $W$ by rendering \textbf{unsolvability} of topology based on the data $\{y_t\}_{t=0}^T$. 
This design is formulated to solve the following problem: 
\begin{subequations}\label{eq:prob_a}
\begin{align}
\textbf{$P_a$}:\quad \text{Find} &~~K\\
\text{s.t.} &~~\lim_{t\to\infty} \| x_{t}- x^*_{\infty} \|_2=0, \label{eq:converge_state} \\
&~~\tilde{W}_{ij}\ge 0~\forall i,j\in\mathcal{V}, \label{eq:row_stochastic} \\
&~~K_{ij}=0~\text{if}~W_{ij}=0~(i\neq j), \label{eq:K_sparse}  \\
& ~~\operatorname{rank}(\tilde{Y}_{0:T-n})<n, \label{eq:require_rank}
\end{align}
\end{subequations}
where $\tilde{Y}_{0:T-n}$ is the alternative of $\tilde{Y}_{0:T-n}^*$ under the dynamics \eqref{eq:local_feedback}. 
The necessities of the above constraints are as follows. 
\begin{itemize}
    \item First, \eqref{eq:converge_state} requires that the consensus convergence \eqref{eq:local_model} is not disrupted, which is essential for the underlying coordination task in the NS. 
    \item Second, \eqref{eq:row_stochastic} maintains the nonnegative weighted-averaging structure of consensus protocols, ensuring that the transmitted message format and computational operations remains unchanged. 
    \item Third, \eqref{eq:K_sparse} guarantees that no new links are introduced in the network, and each node is constrained to interact only with its original neighbors without additional communication overheads. 
    \item Fourth, \eqref{eq:require_rank} ensures that there exist multiple topologies with different spectrum that would lead to the same $\tilde{Y}_{0:T-n}$. 
    Then, the topology cannot be uniquely determined from the data, and the privacy of $W$ is preserved. 
\end{itemize}
\noindent 
We assume that the external observer has no prior knowledge about the feedback design and the network structure. 
Then, the resulting topology privacy is guaranteed by making different $W+K$ undistinguishable in terms of the measured behavior in \textbf{$P_a$}, and thus is independent of specific inference methods employed by observers. 
Extensions of cases where the observer knows the design of $K$ or certain network properties (e.g., sparsity), along with the corresponding topology privacy analysis, are left for future work.

Notice that when the observer can fully access states of all nodes (i.e., $C=I$), 
by following the identifiability definition \eqref{eq:old_estimator}, the constraint \eqref{eq:require_rank} degrades to 
\begin{equation}\label{eq:special_X}
\operatorname{rank}(X_{0:T-1})<n,
\end{equation} 
where $X_{0:T-1}=[x_{0},x_{1},\cdots,x_{T-1}]$. 
When \eqref{eq:special_X} holds, there exist multiple $W$ that would lead to the same $\tilde{X}_{0:T-1}$, which also preserves the privacy of $W$. 
Intuitively, \eqref{eq:require_rank} provides a strong privacy preservation when $C\neq I$, because both the topology $W$ itself and its spectral structure cannot be uniquely determined. 
In contrast, \eqref{eq:special_X} provides a slightly weaker preservation when $C=I$, guaranteeing that only $W$ itself cannot be uniquely determined.

Moreover, even if a worst-case where $\operatorname{rank}(X_{0:T-1})<n$ for $C=I$ is infeasible may occur, we can still provide a last line of privacy preservation design that renders \textbf{inaccuracy} of identifiable topology from the data $\{x_t\}_{t=0}^T$. 
This design is formulated to solve the following problem: 
\begin{subequations}\label{eq:prob_b}
\begin{align}
\textbf{$P_b$}:\quad \text{Find} &~~K\\
\text{s.t.} &~ \eqref{eq:converge_state},~\eqref{eq:row_stochastic},~\text{and}~\eqref{eq:K_sparse}, \\
&~~ W\neq X_{1:T} (X_{0:T-1})^{\dag}, \label{eq:requirement1} 
\end{align}
\end{subequations}
where $X_{0:T-1}$ has full rank and $X_{1:T}=[x_{1},x_{2},\cdots,x_{T}]$. 
The difference of \textbf{$P_b$} from \textbf{$P_a$} is that, here a topology matrix can be uniquely determined from the collected data $X_{0:T}$ while having bias with $W$, as demonstrated in \eqref{eq:requirement1}. 
Since the observer has no knowledge about the design of $K$, it cannot recover the true topology deterministically in the case of \textbf{$P_b$}. 
This can be regarded as the least effects we can achieve to preserve privacy of $W$.

To address the above problems, we first derive the conditions under which the constraints in \textbf{$P_a$} and \textbf{$P_b$} are feasible and characterize the corresponding solution space. 
We then tackle the core challenge of preserving topology privacy in a distributed manner while simultaneously maintaining the original consensus behavior. 
Along this line, we introduce a controllable tradeoff between privacy and convergence performance, and provide distributed design of the required feedback. 
The roadmap of this work is illustrated in Fig.~\ref{fig:roadmap}.

\begin{figure}[t]
\centering
\includegraphics[width=0.48\textwidth]{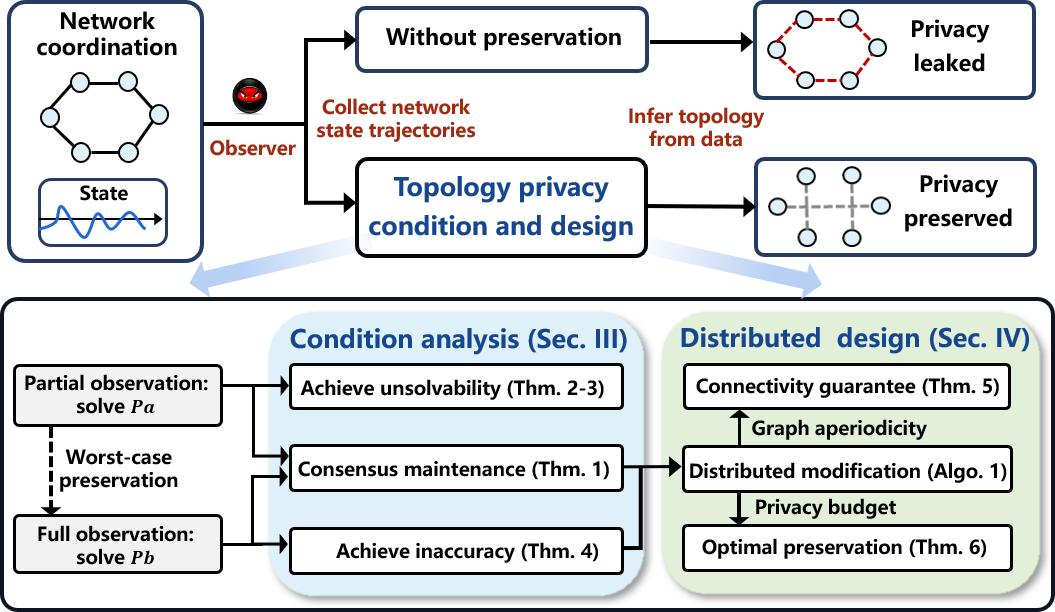}
\vspace{-8pt}
\caption{The roadmap illustration of this work.}
\label{fig:roadmap}
\end{figure}

\section{Privacy Preservation Conditions Analysis}\label{sec:conditions}

In this section, we analyze the conditions of $K$ to meet the privacy and performance constraints. 
Specifically, Sec.~\ref{subsec:unsolve} presents the general conditions of the unsolvability in \textbf{$P_a$}; 
Sec.~\ref{subsec:unobservable} and Sec.~\ref{subsec:enforcing_x0} show the detailed designs to meet the unsolvability in \textbf{$P_a$} from two perspectives; 
Sec.~\ref{subsec:inaccurate} discusses how to meet the inaccuracy in \textbf{$P_b$}. 

First, we give the consensus preservation condition regarding \eqref{eq:converge_state} to support the following analysis. 


\begin{lemma}[see Lemma 1 in \cite{yushan2025Resistant}]\label{th:eigenvalue}
Under Assumption \ref{assu:topo}, the consensus convergence condition \eqref{eq:converge_state} is achieved if and only if $K$ satisfies
\begin{align}
    &K \bm{1}=0, \label{eq:K_c1}\\
    &\pi^\intercal K=0,\label{eq:K_c2}\\
    &|\tilde{\lambda}_i|<1,~i=2,\cdots,n,  \label{eq:K_c3}
\end{align}
where $\pi$ is the left-dominant vector of $W$, and $\tilde{\lambda}_i$ is the $i$-th eigenvalue of $\tilde{W}$. 
\end{lemma}



\begin{lemma}\label{lemma:constraint}
Denote $\mathcal{K}=\{K_0 \in\mathbb{R}^{n\times n}: K_0 \bm{1}=0, ~ \pi^\intercal K_0 =0\}$. 
If $\mathcal{K}$ contains one non-trivial element, then there always exists a non-trivial $K \in \mathcal{K}$, such that the eigenvalue condition \eqref{eq:K_c3} and the non-negativity constraint \eqref{eq:row_stochastic} are satisfied. 
\end{lemma}

The proof of this lemma can be referred to \cite[Theorem 1]{yushan2025Resistant}. 
The key idea is to leverage the continuity of eigenvalues and elements of a matrix under disturbances \cite[Ch 2.1.8]{kato2013perturbation}. 
The above two results reveal that once the homogeneous equalities \eqref{eq:K_c1} and \eqref{eq:K_c2} admit a nontrivial solution,
a matrix $K$ that also meets the eigenvalue modulus constraint \eqref{eq:K_c3} and non-negativity constraint \eqref{eq:row_stochastic} always exists.

\subsection{Conditions for Topology Unsolvability in \textbf{$P_a$}}\label{subsec:unsolve}

In this part, we present the unsolvability conditions for the constraint \eqref{eq:require_rank} in \textbf{$P_a$}. 
We proceed the analysis by introducing invariant subspaces. 
\begin{lemma}{\cite[Ch. 1]{Gohberg2006invariant}} \label{le:condition}
Given a linear transformation $A$ and a vector $x\in \mathbb{R}^{n}$, 
the subspace $\mathcal{A}_\infty = \operatorname{span}\{x,Ax,A^2x,\cdots\}$ is equivalent to 
$$\mathcal{A}_n = \operatorname{span}\{x,Ax,A^2x,\cdots,A^{n-1}x \}.$$ 
Moreover, $\mathcal{A}_n$ is $A$-invariant, i.e., $Ax \in \mathcal{A}_n$ for every vector $x \in \mathcal{A}_n$. 
\end{lemma}


To facilitate the analysis, we define the proper invariant subspace of a matrix $A\in \mathbb{R}^n$ as a linear subspace $\mathcal{S}(A) \subsetneq \mathbb{R}^n$ such that
\begin{equation}
A x \in \mathcal{S}(A), \quad \forall x \in \mathcal{S}(A).
\end{equation}
Here, the term ``proper'' excludes the trivial space $\mathbb{R}^n$, and thus $\dim(\mathcal{S}(A)) < n$. 
Note that such a proper invariant subspace is generally not unique, and we do not distinguish among different choices of $\mathcal{S}(\cdot)$ when no ambiguity arises. 
We then present the following result.



\begin{lemma}\label{le:possibility}
The observation sequence $\tilde{Y}_{0:T-n}$ ($T\ge nm+n-1$) has rank less than $n$ if one of the following conditions holds:
\begin{itemize}
\item there exists a nonzero vector $v\in \mathbb{R}^n$ such that 
\begin{equation}\label{eq:le_sufficient1}
  Q_o v=0,
\end{equation}
where $Q_o=\left[C^\intercal ,(C \tilde{W})^\intercal ,\cdots, (C \tilde{W}^{n-1})^\intercal  \right]^\intercal $ is the observability matrix. 
\item the initial state $x_0$ satisfies 
\begin{align}\label{eq:state_space}
    x_{0}\in\mathcal{S}(\tilde{W}). 
\end{align}
\end{itemize}
\end{lemma}

\begin{proof}
First, notice that the $(t+1)$-th column of $\tilde{Y}_{0:T-n}$ can be equivalently written as $\tilde y_{t:t+n-1}= Q_o x_t$. 
Hence, $\tilde Y_{0:T-n}$ is represented by 
\begin{equation} \label{eq:product_hankel}
  \tilde Y_{0:T-n}= Q_o X_{0:T-n},
\end{equation}
where $X_{0:T-n}=[x_0 , x_1 , \cdots, x_{T-n}]$. 
By rank inequality for matrix multiplication, we have 
\begin{align}
  \operatorname{rank}(\tilde Y_{0:T-n}) &= \operatorname{rank}(Q_o X_{0:T-n})  \nonumber \\
  &\le \min\{\operatorname{rank}(Q_o), \operatorname{rank}(X_{0:T-n}) \},
\end{align}
which indicates there are two cases that guarantee $\operatorname{rank}(\tilde Y_{0:T-n}) < n$. 
For the first case, it is well-known that if $Q_o v=0$ holds for a nonzero $v$, then $\operatorname{rank}(Q_o)<n$. 

As for the other case, recall from Lemma \ref{le:condition} that $\operatorname{rank}(X_{0:T-n}) =\operatorname{rank}(X_{0:n-1}) $ when $T\ge nm+n-1$. 
By referring to \cite[Lemma 3]{yushan2025Resistant}, 
we have
\begin{align}\label{eq:rank_XT}
    \operatorname{rank}(X_{0:n-1}) = n~\text{if and only if}~x_{0}\notin\mathcal{S}(\tilde{W}).
\end{align}
Then, it is straightforward that 
when $x_{0}\in\mathcal{S}(\tilde{W})$, we have $\operatorname{rank}(X_{0:T-n})<n$, and thus $ \operatorname{rank}(\tilde Y_{0:T-n})<n$. 

\end{proof}

Lemma \ref{le:possibility} shows two different ideas to enable $\tilde Y_{p:T-n}$ rank-deficient by designing $K$: 
i) making the modified NS unobservable, 
and ii) enforcing $x_0$ lying in the subspace $S(\tilde{W})$. 
Note that if $x_0 \in \mathcal{S}(W)$, it means that the topology is already unidentifiable from the data in the original dynamics \eqref{eq:local_model}. 
Therefore, the second design idea only focuses on how to render $x_0 \in \mathcal{S}(\tilde{W})$ when $x_0 \notin \mathcal{S}(W)$. 
We will separately present the detailed design of two ideas in the next subsections. 

\subsection{Solve \textbf{$P_a$} by Making the NS Unobservable}\label{subsec:unobservable}
In this part, we show how to ensure the feedback $K$ satisfy all constraints in \textbf{$P_a$} by making the NS unobservable. 

\begin{theorem}\label{th:K_existence}
A non-trivial solution to \textbf{$P_a$} can always be found, if the following condition holds
\begin{align}\label{eq:explicit_con}
\left\{ 
\begin{aligned}
& \operatorname{ker}(W) \cap \operatorname{ker}(C) \neq  \{0\} \\
& |\mathcal{Z}| > 3n-1
\end{aligned}
\right.,
\end{align}
where $\mathcal{Z}=\{W_{ij}: W_{ij} \neq 0, i,j\in\mathcal{V}\}$ and $|\mathcal{Z}|$ is the cardinality of $\mathcal{Z}$. 
\end{theorem}

\begin{proof}
The proof consists of two steps: 
i) we exploit \eqref{eq:le_sufficient1} in Lemma \ref{le:possibility} to find a $K$ satisfying $Q_o v=0$ for some $v$, such that $\operatorname{rank}(\tilde{Y}_{0:T-n})<n$ in \textbf{$P_a$} is satisfied;  
ii) we construct a group of homogeneous equations about $K$ to ensure the satisfaction of the rest constraints in \textbf{$P_a$}.

First, notice that the equation $Q_o v=0$ is equivalent to that all the vectors $\{v, \tilde{W}v, \cdots, \tilde{W}^{T-1}v\}$ lying in $\operatorname{ker}(C)$. 
Let us pick a nonzero $v\in\operatorname{ker}(W) \cap \operatorname{ker}(C)$, whose existence is guanateed by \eqref{eq:explicit_con}, and we have 
\begin{equation}\label{eq:Cv}
Cv=0,\quad Wv=0. 
\end{equation}
We then pick a $K$ such that $Kv = 0$, which implies that 
\begin{equation}
\tilde{W} v=(K+W)v=Kv+Wv=0
\end{equation}
and leads to the remaining vectors $\{(W+K)^t v, t=1,\cdots,T-1 \}$ being all zeros. 
Hence, $Q_o v=0$ is guaranteed and the matrix $\tilde{Y}_{0:T-n}$ will have rank one, satisfying the constraint $\operatorname{rank}(\tilde{Y}_{0:T-n})<n$.

Next, to analyze the solutions to \textbf{$P_a$}, we express the equation $Kv = 0$ and the rest constraints in \textbf{$P_a$} as a linear system of equation
\begin{equation}\label{eq:equations_K}
\left\{
\begin{aligned}
K \bm{1}&=0, \pi^{\intercal} K= 0 \\
K_{ij}&=0,~\text{if}~W_{ij}~\notin \mathcal{Z}\\
Kv&=0
\end{aligned} 
\right.,
\end{equation}
and turn to count the numbers of free variables and independent equations in \eqref{eq:equations_K}. 
From the first equality in \eqref{eq:equations_K}, we have
\begin{equation}\label{eq:Keq}
  \sum_{j=1}^n K_{ij}=0,~i=1,\cdots,n.
\end{equation}
Multiplying the left hand side of \eqref{eq:Keq} by $\pi_i$ and summing over $i$, we have 
\begin{equation}
  \sum_{i=1}^n \pi_i \sum_{j=1}^n K_{ij}=\sum_{i=1}^n \sum_{j=1}^n \pi_i  K_{ij} = \sum_{j=1}^n  (\pi^{\intercal}K)_{j}=0,
\end{equation}
which is exactly a linear combination of the equality constraint $\pi^{\intercal} K = 0$. 
Therefore, the total number of independent linear constraints in the first row of \eqref{eq:equations_K} is at most $2n-1$. 
In the second row of \eqref{eq:equations_K}, there are $n^2-|\mathcal{Z}|$ zero elements in $W$, and it further adds $n^2-|\mathcal{Z}|$ groups of independent equations for $K$. 
The last row of \eqref{eq:equations_K} also contains at most $n$ groups of independent equations. 
Since $K$ has $n^2$ degrees of freedom when without any constraints, by conservatively treating all the above equations for $K$ as independent, a feasible $K$ always exists if the minimal degrees of freedom for the solution to \eqref{eq:equations_K} is larger than zero, i.e.,
\begin{equation}\label{eq:number_solution}
  n^2-(3n-1)-(n^2-|\mathcal{Z}|)>0 \Rightarrow  |\mathcal{Z}|> 3n-1. 
\end{equation}

Finally, since all equations in \eqref{eq:equations_K} are homogenous, according to Lemma \ref{lemma:constraint}, once a nontrivial solution to \eqref{eq:equations_K} is found, 
a feasible solution \textbf{$P_a$} always exists. 
The proof is completed. 
\end{proof}

The proof of Theorem \ref{th:K_existence} is conducted by analyzing solutions to the homogeneous equation \eqref{eq:equations_K}, 
where the second row in \eqref{eq:explicit_con} is a conservative condition requiring that the minimal number of non-zero edges in $W$ is larger than $(3n-1)$. 
In practice, whether $ \operatorname{ker}(W) \cap \operatorname{ker}(C)\neq \{0\}$ holds can be verified by solving the equation $[W^\intercal,~ C^\intercal]v=0 $. 
If the condition \eqref{eq:explicit_con} is satisfied, we can directly solve \eqref{eq:equations_K} to obtain a feasible $K$, which can give multiple solutions. 
For general case where only $ \operatorname{ker}(C) \neq  \{0\}$ holds, we can first select a nonzero $v \in \operatorname{ker}(C) $ and turn to solve the linear feasibility problem:
\begin{equation}\label{eq:equations_K2}
\left\{
\begin{aligned}
&K \bm{1}=0, ~ \pi^{\intercal} K= 0 \\
&K_{ij}=0,~\text{if}~W_{ij}~\in \mathcal{Z} \\
&Kv=-Wv\\
&W_{ij} + K_{ij} \ge 0
\end{aligned} 
\right..
\end{equation}
The solution to \eqref{eq:equations_K2} will automatically satisfy all constraints in \textbf{$P_a$}. 
The above analysis shows that when the matrix $C$ of the NS has no full rank, it provides great potentials to preserve the topology privacy.

\subsection{Solve \textbf{$P_a$} by Enforcing $x_0 \in \mathcal{S}(\tilde{W})$}\label{subsec:enforcing_x0}

This part demonstrates how to enforce $K$ to respect the invariant subspace condition \eqref{eq:state_space}. 
The conclusions in this subsection also apply to the case when $C=I$. 


First, note that enforcing $x_{0}\in\mathcal{S}(\tilde{W})$ is equivalent to enforcing $\operatorname{rank}(X_{0:n-1}) <n$. 
Algebraically, this rank deficiency can be translated via Cayley-Hamilton theorem \cite[Th. 2.4.3.2]{horn2012matrix} into the existence of a nontrivial matrix polynomial $P(\tilde W)=0$, which introduces additional free polynomial coefficients besides the unknown $K$. 
Furthermore, the row-stochastic condition $\tilde W\mathbf 1=\mathbf 1$ couples these coefficients through an extra linear relation. 
As a result, each entry of $P(\tilde W)$ becomes a high-order multivariate polynomial in the entries of $K$.  
Therefore, solving this highly nonlinear problem does not as such lead to a general tractable method for finding an appropriate $K$.



However, in some special cases, we can design the feedback $K$ in an explicit way. 
For simplicity of illustration, we first that consider that $W$ is diagonalizable and provide extension to non-diagonalizable case  later. 
Then, the eigendecomposition of diagonalizable $W$ is given by 
\begin{equation}\label{eq:decomposition}
W=M\operatorname{diag}\{\lambda_1,\cdots,\lambda_n\} M^{-1}=\sum\nolimits_{i = 1}^n {{p_i}q_i^{\intercal}{\lambda _i}},
\end{equation}
where $\lambda_i$ is the $i$-th eigenvalue of $W$, $M=[p_1,\cdots,p_n]$ and $(M^{-1})^{\intercal}=[q_1,\cdots,q_n]$ are the transformation matrices consisting of the right and left eigenvectors of $W$, respectively. 
The eigenvectors satisfy
\begin{equation}
  {q_i^{\intercal}}{p_i}=1,~{q_i^{\intercal}} p_j =0~(i\ne j).
\end{equation}
Specifically, $q_1=\pi$ and $p_1=\bm{1}$ under Assumption \ref{assu:topo}. 
Then, the initial state $x_0$ can be expanded in the eigenvectors $\{p_1,\cdots, p_n\}$, given by 
\begin{equation}\label{expansion:x_0}
  x_0= \sum\nolimits_{i=1}^n \beta_i p_i, 
\end{equation}
where $\beta_1=q_1^{\intercal} x_0$ is exactly the original consensus value. 
Before presenting the next result, we define an index set $\mathcal{I}= \{1,2,\cdots, n\}$, and denote $\operatorname{supp}(\cdot)$ as the index set of nonzero elements of a matrix.

\begin{theorem}\label{th:construct_K}
If there exists a subset $\mathcal{I}_s \subseteq \mathcal{I}\setminus \{1\}$ such that 
\begin{equation}\label{eq:cond_eig}
 |\mathcal{I}_s|\ge 2,\quad  \operatorname{supp}\left\{ \sum\nolimits_{i \in \mathcal{I}_s} \lambda_i p_i q_i^{\intercal}  \right\} \subseteq \operatorname{supp}\{ W\},
\end{equation}
then the following constructed $K$ satisfies the constraints of convergence \eqref{eq:converge_state}, sparsity \eqref{eq:K_sparse}, and rank \eqref{eq:require_rank}:  
\begin{align}\label{eq:K_construct}
K_c=\sum\nolimits_{i \in \mathcal{I}_s} -\lambda_i p_i q_i^{\intercal}. 
\end{align}
Moreover, if $W+K_c\ge0$ holds element-wisely, then $K_c$ is a solution to \textbf{$P_a$}. 
\end{theorem}

\begin{proof}
In this proof, we exploit \eqref{eq:state_space} in Lemma \ref{le:possibility} to construct the feedback as \eqref{eq:K_construct}, and show that it satisfies the three constraints \eqref{eq:converge_state}, \eqref{eq:K_sparse}, and \eqref{eq:require_rank} in \textbf{$P_a$}.
First, when $\operatorname{supp}\left\{ \sum\nolimits_{i \in \mathcal{I}_s} \lambda_i p_i q_i^{\intercal}  \right\} \subseteq \operatorname{supp}\{ W\}$ holds, it implies that 
\begin{equation}
  W_{ij} \neq 0,~\forall \{i,j\}\in \operatorname{supp}\left\{ \sum\nolimits_{i \in \mathcal{I}_s} \lambda_i p_i q_i^{\intercal}  \right\}. 
\end{equation}
If $\left [\sum\nolimits_{i \in \mathcal{I}_s} \lambda_i p_i q_i^{\intercal} \right ]_{ij} = 0$, it does not matter whether $W_{ij}$ is zero. 
Hence, the local communication requirement \eqref{eq:K_sparse} is satisfied. 
For the convergence constraint \eqref{eq:converge_state}, we just need to focus on the equalities \eqref{eq:K_c1} and \eqref{eq:K_c2} by Lemma \ref{lemma:constraint}. 
Then, the construction of \eqref{eq:K_construct} satisfies that 
\begin{align}
&K_c\bm{1}= \sum\nolimits_{i \in \mathcal{I}_s} -\lambda_i p_i (q_i^{\intercal} \bm{1}) = 0,  \\
&q_1^{\intercal}K_c= \sum\nolimits_{i \in \mathcal{I}_s} - \lambda_i (q_1^{\intercal} p_i) q_i^{\intercal} = 0,
\end{align}
where we use $ q_i^{\intercal} \bm{1}=0 $ and $q_1^{\intercal} p_i=0$ for all $i\in \mathcal{I}_s $. 
Meanwhile, by multiplying $(W+K_c) $ with the eigenvectors $\{p_1,p_2,\cdots, p_n\}$ of $W$,
we have 
\begin{equation}\label{eq:Wp_property}
(W+K_c) p_j\!=W p_j - \!\! \sum_{i \in \mathcal{I}_s} \lambda_i p_i (q_i^{\intercal} p_j)\!=\left\{
\begin{aligned}
&0, &&\text{if}~j\in \mathcal{I}_s \\
&\lambda_j p_j, &&\text{otherwise} 
\end{aligned}
\right. . 
\end{equation}
It is clear that $(W+K_c)$ has a unique eigenvalue $1$ corresponding to eigenvector $\bm{1}$, and its other eigenvalues are either the same as those of $W$ or zero. 

Next, for the rank-deficient condition \eqref{eq:require_rank}, we will first prove that $x_t \in \mathcal{S}_p,~\forall t\ge1$ and then show $\operatorname{rank}[x_0,x_1,\cdots,x_{n-1}]\le n-1$ when $|\mathcal{I}_s|\ge 2$. 
Denote $\mathcal{S}_p\!=\operatorname{span}\{p_i, i\!\in\! \mathcal{I}\setminus \mathcal{I}_s \}$, and it follows from \eqref{expansion:x_0} that 
\begin{align}
x_1&=(W+K_c) x_0 = \sum\nolimits_{i\in \mathcal{I}} (W+K_c) \beta_i p_i  \nonumber \\
   &= \sum\nolimits_{i\in \mathcal{I}\setminus \mathcal{I}_s }\beta_i \lambda_i p_i ~\in \mathcal{S}_p,
\end{align}
where the last row utilizes the property \eqref{eq:Wp_property}. 
Similarly, the following state sequences satisfy
\begin{align}
&x_2= (W+K_c) x_1 = \!\! \sum_{i\in \mathcal{I}\setminus \mathcal{I}_s } \left( W+K_c \right) \beta_i \lambda_i p_i = \!\! \sum_{i\in \mathcal{I}\setminus \mathcal{I}_s } \beta_i \lambda_i^2 p_i  \nonumber \\ 
 &\Rightarrow ~~  x_t= \sum_{i\in \mathcal{I}\setminus \mathcal{I}_s }  \beta_i \lambda_i^t p_i \in \mathcal{S}_p,~t\ge1. 
\end{align}
Notice that when $|\mathcal{I}_s|\ge 2$, $\operatorname{dim}(\mathcal{S}_p)\le n-2$ by definition. 
Considering that $x_0$ is not in the subspace $\mathcal{S}_p$, we further have
\begin{align}
	\operatorname{rank}[x_0,x_1,\cdots,x_{n-1}] &\le \operatorname{dim}( \operatorname{span}\{x_0\}+\mathcal{S}_p) \nonumber \\
	&\le 1+\operatorname{dim}(\mathcal{S}_p) \le n-1. 
\end{align}
In summary, the constructed $K_c$ satisfies the constraints \eqref{eq:converge_state}, \eqref{eq:K_sparse}, and \eqref{eq:require_rank}, and it is a solution to \textbf{$P_a$} if $W+K\ge0$ further holds. 
The proof is completed. 
\end{proof}

Compared with directly solving a high-dimensional nonlinear matrix equation, 
Theorem \ref{th:construct_K} provides an explicit spectrum-driven construction for $K$. 
The key point lies in removing a small number of non-consensus eigenmodes $\{ \lambda_i p_i q_i^{\intercal} , i \in \mathcal{I}_s\}$, which are allowed to cancel each other on zero entries of $W$, and thus guaranteeing the NS to evolve in the proper invariant subspace $\mathcal{S}_p$. 
Denote $ H_{\ell i}= \lambda_i p_i(r_\ell) q_i(c_\ell)$ given the $(r_\ell, c_\ell)$-th zero entry of $W$, and construct $H=[H_{\ell i}]_{\ell=1,\cdots, |\mathcal{Z}| }^{i=1,\cdots,n}$. 
Then, finding the subset $\mathcal{I}_s$ is equivalent to finding a vector $v_s\in\{0,1\}^n$ satisfying $H v_s = 0$, 
which is a typical $0$-$1$ linear feasibility problem.  

Note that the nonnegativity constraint \eqref{eq:row_stochastic} is not enforced in the constructed $K_c$.
Instead, this explicit design allows the key structural properties (namely consensus preservation, rank deficiency, and sparsity) to be established analytically, 
and leaves \eqref{eq:row_stochastic} to a simple \emph{a posteriori} verification or refinement. 
If $W+K_c\ge 0$ is not satisfied, we can further seek an additive refinement $\Delta K$ from the following linear feasibility problem:
\begin{subequations}\label{eq:deltaK_feas}
\vspace{-10pt}
\begin{align}
\text{find}~~ &\Delta K \\
\text{s.t.}~~& W+K_c+\Delta K \ge 0 , \\
& \Delta K\bm{1}=0,~\pi^\intercal \Delta K=0, \label{eq:f_con} \\
& \operatorname{supp}(\Delta K)\subseteq \operatorname{supp}(W), \label{eq:f_spa}\\
& \Delta K p_i = 0,\quad \forall i\in\mathcal{I}_s, \label{eq:f_rank}
\end{align}
\end{subequations}
where \eqref{eq:f_con} ensures the consensus convergence and \eqref{eq:f_spa} maintains the sparsity pattern. 
More importantly, \eqref{eq:f_rank} prevents introducing the eigenmodes removed by $K_c$, and thus preserves the same invariant subspace $\mathcal{S}_p$ responsible for the rank-deficiency guarantee. 
Consequently, when \eqref{eq:deltaK_feas} is feasible, any feasible solution yields a final matrix $K=K_c+\Delta K$
that satisfies all constraints in \textbf{$P_a$}. 
If $W$ is not diagonalizable, one can still work with its generalized eigenvectors. 
The details are limited here due to the space limit.

\subsection{Conditions for Topology Inaccuracy in \textbf{$P_b$}}\label{subsec:inaccurate}

In this subsection, we briefly analyze the fully observed case (i.e., $C=I$).
While the previous subsection preserves topology privacy by making $X_{0:T-1}$ rank-deficient, the required conditions may fail in practice. 
To address this, we provide an alternative design via \textbf{$P_b$}, which intentionally biases the topology identifiable from $\{x_t\}_{t=0}^{T}$, 
thereby serving as a last line of preservation for topology privacy.

First, the consensus-preserving constraints \eqref{eq:K_c1}--\eqref{eq:K_c2} can be written compactly in vector form as
\begin{equation}\label{eq:M_constraint}
\left\{\begin{aligned}
(\bm{1}_n^\intercal \otimes I_n )\operatorname{vec}(K)=0\\
(I_n \otimes \pi^\intercal)  \operatorname{vec}(K)=0
\end{aligned}
\right.
\Rightarrow  M \theta_K=0,
\end{equation}
where $M=[(\bm{1}_n^\intercal \otimes I_n )^\intercal, (I_n \otimes \pi^\intercal)^\intercal]^\intercal \in \mathbb{R}^{2n \times n^2}$ and $\theta_K = \operatorname{vec}(K)  \in \mathbb{R}^{n^2}$. 
Recall that $\mathcal{Z}$ is the candidate edge set in \eqref{eq:explicit_con} with $z=|\mathcal{Z}|$, and let $\tilde{\theta}_K\in\mathbb{R}^{z}$ collect the nonzero entries of $\theta_K$ allowed by $\mathcal{Z}$. 
Removing the corresponding zero columns from $M$ yields $\tilde M\in\mathbb{R}^{2n\times z}$, and we have
\begin{equation}\label{eq:new}
\tilde M\,\tilde{\theta}_K=0.
\end{equation}
We next provide existence conditions for solutions to \textbf{$P_b$}.
\begin{theorem}\label{th:exist_K}
A non-trivial solution to \textbf{$P_b$} can always be found, if the following condition holds
\begin{equation}\label{eq:condition}
\operatorname{rank}(\tilde{M})< z.  
\end{equation}
\end{theorem}
   
\begin{proof}
First, note that the homogeneous equation \eqref{eq:new} has non-trivial solutions only when $\operatorname{rank}(\tilde{M})<z$. 
Under this condition, the augmented version of $\tilde{\theta}_{K^*}$, denoted by $\theta_{K^*}$, is obtained by inserting zeros at all positions specified by the sparsity constraint \eqref{eq:K_sparse}. 
We restore the solution matrix $K^*$ by reshaping $\theta_{K^*}$ back to its matrix form according to the vectorization rule used in \eqref{eq:M_constraint}. 
As a result, the reconstructed $K^*$ automatically meets the equality conditions \eqref{eq:K_c1} and \eqref{eq:K_c2}.  
Finally, since $K^*$ is nontrivial, it follows from Lemma \ref{lemma:constraint} that there always exists a $K$, such that $|\tilde{\lambda}_i|<1$ ($i=2,\cdots,n$) and $W_{ij}+K_{ij}\ge 0$ holds, 
which completes the proof. 
\end{proof}

Theorem \ref{th:exist_K} demonstrates the existence of a feasible feedback matrix to make the topology identifiable from the observations inaccurate with respect to the true one. 
The key lies in that the sparsity structure of $W$ can allow more free parameters than the number of independent constraints. 
A straightforward way of obtaining an appropriate $K$ for \textbf{$P_b$} is to 
\begin{itemize}
\item obtain the basis of the kernel space of $\tilde{M}$;
\item combine the basis vectors to form a candidate $K^*$ given desired criteria (e.g., maximizing $\|K^*\|_0$); 
\item scale $K\!=\!\epsilon K^* (\epsilon>0)$ such that $|\tilde{\lambda}_i|\!<\!1,i=2,\!\cdots\!,n$.
\end{itemize}



\begin{remark}
In summary, the solving manners to meet constraints in \textbf{$P_a$} and \textbf{$P_b$} rely on the knowledge of the communication topology. 
This paradigm of centrally designing the feedback $K$ while enabling distributed execution is well aligned with many practical scenarios, where a central unit computes global parameters offline or online, and disseminates parts of them to individual nodes for local coordination.  
Examples include wireless sensor networks with a sink \cite{7004042}, power networks with area control centers for distributed secondary control \cite{8892668}, and multi-robot teams with centralized task computations and decentralized updates \cite{alonso2019distributed}. 
\end{remark}

Nevertheless, if the NS strictly excludes a central unit and constrains the nodes to communicate in a totally distributed manner, 
making each node locally coordinate to find the optimal feedback $K$ without global information is significantly hard (especially when targeting at the unsolvability). 
In the next section, we will demonstrate how to achieve a distributed topology preservation design.

\section{Performance Tradeoff and Distributed Preservation Design}\label{sec:designs}

To achieve the topology privacy preservation in a tractable distributed manner, 
this section seeks for a tradeoff with the zero deviation constraint \eqref{eq:converge_state} in \textbf{$P_b$},  
and proposes a distributed design under privacy budget to enforce the topology inaccuracy from available data.

\subsection{Laplacian based Design}\label{subsec:Laplacian}

Considering that $W$ itself exhibits the desired structure for solutions to \textbf{$P_b$}, our preliminary work \cite{yushan2025Resistant} proposed a Laplacian based design 
\begin{equation}\label{eq:local_K}
K_{L,ij}=\left\{\begin{aligned}
&\alpha W_{ij},&&\text{if}~i\neq j \\
&\alpha (W_{ii}-1),&&\text{if}~i=j
\end{aligned}\right.,
\end{equation}
where $\alpha > 0$ a positive scalar gain. 
It is clear that by \eqref{eq:local_K}, node $i$ only use information from its neighbors and itself. 
In a global form, $K_L$ can be written as 
\begin{equation}\label{eq:K_design}
K_L = -\alpha (I-W),
\end{equation}
which can be regarded a scaled negative Laplacian of the graph $\mathcal{G}$. 
Specifically, it is proved in \cite{yushan2025Resistant} that all constraints in \textbf{$P_b$} can be met if the parameter $\alpha$ satisfies
\begin{equation}\label{eq:alpha_bound}
\alpha < \frac{1 - r_{\max}}{1 + r_{\max}}, 
\end{equation}
where $r_{\max}=\max\{|\lambda_i|,i=2,\cdots,n\}$. 

Despite its simple and efficient structure, the above design has a non-negligible drawback: 
the privacy of the weight signs is not preserved, and the off-diagonal weights are identifiable from the data up to a scalar ambiguity (i.e., the $\alpha$ in \eqref{eq:local_K}). 
Such information is very sensitive in many applications and should be taken into account for the preservation design.


\subsection{The Proposed Distributed Topology Preservation Method}

Recall that Assumption \ref{assu:topo} specifies the strong connectivity of the NS to ensure the consensus. 
Before we show how this assumption can be relaxed for our method, 
we refer to \cite[Chap. 3]{FB-LNS} and introduce the following graph definitions needed in this section. 
A directed graph is periodic if the greatest common divisor of the lengths of all its simple cycles is larger than one. 
Any strongly-connected digraph with a self-loop is aperiodic. 
We then present the following result. 

\begin{lemma}[Theorem 5.1 in \cite{FB-LNS}]\label{le:graphical}
For a row-stochastic matrix $\tilde{W}$ and its associated digragh $\tilde{\mathcal{G}}$, 
if $\tilde{\mathcal{G}}$ contains a root node and the subgraph of root nodes is aperiodic, then the NS under the dynamics \eqref{eq:local_feedback} can reach to the consensus point 
\begin{align}
    \lim_{t\to\infty} x_t = (\tilde{\pi}^\intercal x_0) \bm{1},
\end{align}
where $\tilde{\pi}$ is the left-dominant eigenvector of $\tilde{W}$. 
\end{lemma} 

Based on Lemma \ref{le:graphical}, our key idea is to modify the local topology such that the global topology satisfies the graphical condition in Lemma \ref{le:graphical}. 
However, note that when there is no special design on the left-dominant eigenvector of $\tilde{W}$, modifying the topology can easily cause a deviation between $\tilde{\pi}$ and $\pi$ (consequently $x_{\infty}$ and $x^*_{\infty}$). 
Specifically, the deviation $\| \pi - \tilde{\pi} \|_1$ can be upper bounded by \cite{cho2001comparison}
\begin{align}\label{eq:tradeoff}
\| \pi - \tilde{\pi} \|_1 \le \| W^{\#}\|_{\infty}\|K\|_{\infty},
\end{align}
where $W^{\#}$ is the group inverse\footnote{The group inverse of matrix $W$ is the unique square matrix $W^{\#}$ satisfying $W W^{\#} W = W$, $W^{\#} W W^{\#}= W^{\#}$, and  $W^{\#} W = W W^{\#}$.} of $W$. 
It is clear from \eqref{eq:tradeoff} that $\| \pi - \tilde{\pi} \|_1$ is upper bounded by a term determined by $\|K\|_{\infty}$. 
This point motivates us to control $\| \pi - \tilde{\pi} \|_1$ by constraining the rows of $K$ separately. 
Then, we introduce a privacy budget $\tau$ for modifying each row, described by
\begin{align}\label{eq:budget_K}
\|K_{[i,:]}\|_1=&\sum\nolimits_{j=1}^n |K_{ij}| <\tau,~i=1,\cdots,n.
\end{align}
As long as \eqref{eq:budget_K} holds, it ensures that $\| \pi - \tilde{\pi} \|_1 \le \| W^{\#}\|_{\infty} \tau $. 

Based on the above formulation, our goal is to hide as many existing edges as possible while respecting the constraint: 
\begin{align}\label{eq:common_K}
K_{ij}=0~\text{if}~W_{ij}=0,\quad \sum\nolimits_{j=1}^n K_{ij}=0.
\end{align}
We next present a maximum-beacon based method.

\noindent\textbf{Step 1}: Each node maintains a beacon variable $b_i$ and a depth variable $d_i$, initialized by 
\begin{equation}
  b_i^{(0)}=\xi_i,~d_i^{(0)}=0,
\end{equation}
where $\xi_i$ is an independent random variable drawn from a continuous distribution. 
This initialization ensures that the nodes do not have the same initial
beacon almost surely, i.e., 
\begin{align}
\Pr\left( \exists i\neq j: b_i(0)=b_j(0), ~i,j\in\mathcal{V} \right)=0.
\end{align}
\noindent\textbf{Step 2}: Synchronously update the $\{x_i,b_i,d_i\}$ at iteration $t=0,1,\cdots,n-2$, by 
\begin{align}
\!\!x^i_{t+1}&\!=\sum_{j=1}^n W_{ij} x^j_{t}, ~b_i^{(t+1)} \!=\! \max\!\Big\{ b_i^{(t)}, \max_{j\in\mathcal{N}_i} b_j^{(t)} \Big\}, \label{eq:bupdate}\\
\!\! d_i^{(t+1)} &\!\!= \!
\begin{cases}
0, \quad \quad\quad \quad\quad\quad\quad\quad\quad\quad\text{if } b_i^{(t+1)}=\xi_i\\
1+\!\min \{d_j^{(t)}\!:\! j\!\in\!\mathcal{N}_i, b_j^{(t)}\!=\!b_i^{(t\!+\!1)}  \},  \text{otherwise}
\end{cases}\!\!\!\!\!\!.
\label{eq:d-update}
\end{align}

\noindent\textbf{Step 3}: At iteration $n-1$, set the local feedback for nodes. 
\begin{itemize}
\item If $b_i^{(n-1)}=\xi_i$ (new root node), select node $i$ itself and a parent node $j_0\in \arg\max_{ j\in\mathcal{N}_i } d_j^{(n-1)}$ 
and design $K_{[i,:]}$ by solving the following problem
\begin{subequations}\label{eq:root}
\begin{align}
\min_{K_{[i,:]}} &~~\| W_{[i,:]}+K_{[i,:]} \|_0 \\
\text{s.t.} ~~& \eqref{eq:common_K}~\text{and}~\eqref{eq:budget_K},\\
&W_{ii}+K_{ii}>0,~W_{ij_0}+K_{ij_0}>0.  \label{eq:root_constraint} 
\end{align}
\end{subequations}
\item If $b_i^{(n-1)}\neq \xi_i$ (not root node), fix a parent node $j_p\in \arg\min_{ j\in\mathcal{N}_i } d_j^{(n-1)} $ and design $K_{[i,:]}$ by solving the following problem
\begin{subequations}\label{eq:nroot}
\begin{align}
\min_{K_{[i,:]}} &~~\| W_{[i,:]}+K_{[i,:]} \|_0 \\
\text{s.t.} ~~& \eqref{eq:common_K}~\text{and}~\eqref{eq:budget_K},\\
&W_{i j_p}+K_{ij_p}>0. \label{cons:nroot_b}
\end{align}
\end{subequations}
\item Hereafter, update states by $x^i_{t+1}=\sum_{j=1}^n (W_{ij}+K_{ij}) x^j_{t}$. 
\end{itemize}

The key idea of Step 1 is to leverage the randomized initialization and render that the maximum beacon value is attained at a unique node almost surely. 
Step 2 combines the classic maximum consensus protocol and breadth-first search (BFS) tree construction \cite[Chap. 5]{hirvonen2020distributed} to update $\{b_i,d_i\}$. 
The node with maximum beacon is treated as the new root (it is sure that one can adopt other criteria to determine such a node). 
Since only the maximum beacon will not be updated during the iterations, the new root can be directly determined if $b_i^{(n-1)}=\xi_i$ holds. 
The intuition of adopting different parent node criterion in Step 3 lies in two aspects. 
\begin{itemize}
\item  For a non-root node $i$, the parent selection rule $j_p \in \arg\min_{j\in\mathcal{N}_i} d_j^{(n-1)}$ ensures that 
node $i$ maintains at least one in-neighbor closer to the root, 
thereby prompting a directed spanning tree oriented toward the root. 
\item For the root node, the parent is selected as $j_0 \in \arg\max_{j\in\mathcal{N}_i} d_j^{(n-1)}$, 
such that the root maximizes its coverage and strengthens the
propagation of information throughout the network. 
\end{itemize}
The detailed analysis is given in the next subsection.

\subsection{Analysis of the Proposed Method}

First, denote the node with maximum beacon value as 
\begin{equation}
r^*=\arg\max_{i\in \mathcal{V}} b_i^{(0)}, 
\end{equation}
and denote the graph associated with $\tilde{W}$ as $\tilde{\mathcal{G}}$.  
Then, we show that the aperiodicity of $\tilde{\mathcal{G}}$ is preserved by the proposed method.

\begin{theorem}\label{eq:new_conver}
By the proposed maximum-beacon based method, with probability one, the graph $\tilde{\mathcal{G}}$ contains a root node and the subgraph of root nodes is aperiodic. 
\end{theorem}

\begin{proof}
First, the randomized initialization for the beacons in Step 1 ensures that all $b_i^{(0)}$ are different with probability one. 
Hence, node $r^*$ is unique with probability one, and the subsequent deterministic analysis are all built upon it. 
We then proceed the proof from the following two parts. 

\emph{i) $\tilde{\mathcal{G}}$ contains a root node.} 
In Step 2, the updating rule for $b_i$ is a max-consensus process. 
Since $\mathcal{G}$ is originally strongly connected, there always exists a directed path from $r^*$ to all other nodes, where the path length is at most $n-1$. 
Let $\operatorname{dist}(r^*,i)$ be the shortest path length from $r^*$ to $i\neq r^*$. 
Since the update rule for $d_i$ is a BFS-type distance labeling rooted at $r^*$, after $n-1$ iterations we have \cite[Chap. 5]{hirvonen2020distributed}
\begin{equation}\label{eq:d_n_1}
  d_i^{(n-1)}=\operatorname{dist}(r^*,i). 
\end{equation}
In Step~3, for a node $i\neq r^*$, the parent node $j_p\in \arg\min_{ j\in\mathcal{N}_i } d_j^{(n-1)}$ is fixed and the constraint $W_{ij_p}+K_{ij_p}>0$ in~\eqref{eq:nroot} enforces an arc $j_p \rightarrow i$ in $\tilde{\mathcal G}$. 
Meanwhile, note that the distance function $\operatorname{dist}(r^*,i)$ satisfies the Bellman optimality relation on the directed graph, i.e., for any $i\neq r^*$,
\begin{equation}\label{eq:bellman_dist}
\operatorname{dist}(r^*,i)=1+\min_{j\in\mathcal N_i}\operatorname{dist}(r^*,j).
\end{equation}
Applying the above property to $i$ and $j_p$, we have
\begin{equation}\label{eq:parent_plus_one}
d_i^{(n-1)}=1 + d_{j_p}^{(n-1)}.
\end{equation}
Recursively, the parent arcs construct a directed path from $r^*$ to $i$. 
By arbitrariness of $i$, the node $r^*$ is a root node in $\tilde{\mathcal G}$.

\emph{ii) The subgraph of root nodes is aperiodic.} 
In Step 3, when solving \eqref{eq:root}, the root node $r^*$ has a predecessor node $j_0$ due to the constraint $W_{ij_0}+K_{ij_0}>0$. 
Thus, there exists a subset $\mathcal{V}_c\subseteq \mathcal{V}$ such that the nodes in $\mathcal{V}_c$ form a directed cycle $r^* \rightarrow \cdots \rightarrow j_0 \rightarrow r^* $. 
Clearly, $\mathcal{V}_c$ is a strongly connected component in which all nodes are global root roots of $\tilde{\mathcal{G}}$. 
Meanwhile, since $W_{ii}+K_{ii}>0$ is enforced, the greatest common divisor of the lengths of directed cycles in $\mathcal{V}_c$ is exactly 1, which implies that $\mathcal{V}_c$ is aperiodic. 
The proof is completed. 
\end{proof}

Theorem \ref{eq:new_conver} essentially reveals that the network connectivity and the convergence of $\tilde{W}$ are guaranteed by the proposed method. 
A drawback here is that the strongly connected characteristic of $\mathcal{G}$ is not necessarily preserved.  
However, there are at least two root nodes in $\tilde{\mathcal{G}}$, and the selection rule $j_0\in \arg\max_{ j\in\mathcal{N}_i } d_j^{(n-1)}$ is used to include as many roots as possible. 
If there exists a $j\in \mathcal{N}_{r^*}$ such that the shortest path length from $r^*$ to $j$ is $n-1$, then $\tilde{\mathcal{G}}$ is still strongly connected.

We now analyze the solutions of problem \eqref{eq:root} and problem \eqref{eq:nroot}. 
Intuitively, finding the minimal $\ell_0$-norm $\| W_{[i,:]} + K_{[i,:]}\|_0$ is a mixed integer linear programming (MILP) problem, which is NP-hard in general. 
Thanks to the structural constraint \eqref{eq:common_K}, the optimal solutions of problem \eqref{eq:root} and \eqref{eq:nroot} can be constructed in polynomial time. 
Define the mandatory neighbor set $\mathcal{N}_i^m$ and the non-mandatory candidate neighbor set $\mathcal{C}_i^m$ for node $i$ as 
\begin{equation}\label{eq:candidate}
\mathcal{N}_i^m=\left\{
\begin{aligned}
  &\{j_p\},&\text{if}~i\neq r^*\\
  & \{ {j_0},{i}\},&\text{otherwise}
\end{aligned}\right. ,\quad
\mathcal{C}_i^m=\mathcal{N}_i^+ \setminus \mathcal{N}_i^m. 
\end{equation} 
Let $W_{ij_{(1)}} \ge W_{ij_{(2)}}\ge \cdots \ge W_{ij_{(|\mathcal{C}_i^m|)}}$ be the sorted weights of nodes in $\mathcal{C}_i^m$. 
We then present a theorem that characterizes the explicit solutions of the two problems. 
This is completed by a construction-based proof, which will be further used to design the specific algorithm that solves the two problems.


\begin{theorem}\label{th:solution}
For the case $i\neq r^*$, the optimal value of the problem \eqref{eq:nroot} is given by 
\begin{align}\label{eq:c1_cost}
c_{nroot}^*=\begin{cases}
1, & \tau \ge 2(1-W_{ij_p})\\
1+\bar c_1, & \text{otherwise}
\end{cases},
\end{align}
where
$\bar c_1 =\min_c
\left\{c: \sum_{r=1}^{c} W_{i j_{(r)}}
\ge 1-\frac{\tau}{2}-W_{ij_p}
\right\}$, 
and the corresponding optimal support sets are $\{ j_p\}$ and $\{ j_p,j_{(1)},\cdots,j_{(c_{nroot}^*-1)} \}$, respectively. 
For the case $i= r^*$, the optimal value of the problem \eqref{eq:root} is given by 
\begin{align}\label{eq:c2_cost}
c_{root}^*=
\begin{cases}
2, & \tau \ge 2(1-W_{ii}-W_{ij_0})\\
2+\bar c_2, & \text{otherwise}
\end{cases},
\end{align}
where
$
\bar c_2 \!=\! \min_c
\left\{c:\sum_{r=1}^{c} W_{i j_{(r)}}
\!\ge 1 \!-\frac{\tau}{2} \! - W_{ii} \!- W_{ij_0}\right\}
$, 
and the corresponding optimal support sets are $\{i,j_0\}$ and $\{i,j_0,j_{(1)},\ldots,j_{(c_{root}^*-2)}\}$, respectively.
\end{theorem}

\begin{algorithm}[t]
  \caption{Heuristic based local topology modification}
  \label{alg:row_design}
  \begin{algorithmic}[1]
    \REQUIRE local topology $W_{[i,:]}$, mandatory neighbor set $\mathcal{N}_i^m$, budget $\tau>0$, and a small adjusting parameter $\delta>0$. 
    \ENSURE local feedback $K_{[i,:]}$.
    \STATE Initialize $K_{[i,:]}=\bm{0}_{1\times n}$. 
    \STATE Obtain non-mandatory candidate set $\mathcal{C}_i^m$ by \eqref{eq:candidate}, and sort indices in $\mathcal{C}_i^m$ in ascending order of $W_{ij}$. 
    \STATE Set auxiliary variables $\mathcal{N}_i^d  = \emptyset$ and $a = 0$. 
    \FOR{each $W_{ij} \in \mathcal{C}_i^m$ in ascending order}
      \IF{$2(a + W_{ij}) \le \tau$}
        \STATE Update $\mathcal{N}_i^d  = \mathcal{N}_i^d  \cup \{j\}$ and $a =a + W_{ij}$.
      \ELSE
        \STATE \textbf{break}
      \ENDIF
    \ENDFOR
    \STATE Set $K_{ij}= -W_{ij}, \forall j\in\mathcal{N}_i^d $, compute $\alpha_z=\sum_{j\in\mathcal{N}_i^d } W_{ij}$. 
    \STATE Obtain remaining candidate set $ \mathcal{C}_r=\{\mathcal{C}_i^m \setminus \mathcal{N}_i^d  \} \cup \mathcal{N}_i^m$. 
    \STATE Set $\Delta_1 =\frac{\alpha_z}{|\mathcal{C}_r|}$, and compensate the row-sum condition by updating $K_{ij} = K_{ij} + \Delta_1$ for all $j\in\mathcal{C}_r$. 
    \IF{$|\mathcal{C}_r|=1$}
      \STATE return $K_{[i:,]}$. 
    \ENDIF
    \STATE Compute left budget $\tau_r\!=\!\left\{ \begin{aligned} &\min\{\tau \!-\! 2\alpha_z, 1 \}, ~\text{if}~|\mathcal{C}_r|\!=\!2 \\  &\tau \!- 2\alpha_z, ~\text{otherwise} \end{aligned}\right.$. 
    \STATE Construct $\mathcal{C}_a$/$\mathcal{C}_b$ as the first/last $h=\lfloor \frac{|\mathcal{C}_r|}{2} \rfloor $ indices in $\mathcal{C}_r$
    \STATE Set $\Delta_2=\min \{ \frac{\tau_r}{2h}, \min_{j\in\mathcal{C}_r} (W_{ij} + K_{ij} - \delta) \}$ for non-negativity. 
    \STATE Update $K_{ij} = K_{ij} + \Delta_2$ for all $j\in\mathcal{C}_a$, $K_{ij} = K_{ij} - \Delta_2$ for all $j\in\mathcal{C}_b$. 
    \STATE If $K_{ij_1}=0$ happens for a $j_1\in\mathcal{C}_b$, adjust $K_{ij_1} = K_{ij_1} + \delta$ along with $K_{ij_2} = K_{ij_2} - \delta$ for a $j_2\in\mathcal{C}_a$. 
  \end{algorithmic}
\end{algorithm}

\begin{proof}
First, we consider node $i\neq r^*$ and prove \eqref{eq:c1_cost}. 
Due to the constraint \eqref{cons:nroot_b}, node $i$ has at least one in-neighbor $j_p\in\mathcal{N}_i$. 
Suppose a single edge is allowed on $\tilde{W}_{[i,:]}$ ($\tilde{w}_{i j_p}=1$), and denote $e_{j_p}$ as the unit row vector with the $j_p$-th element being $1$. 
Then, the deviation $\| W_{[i,:]}-e_{j_p} \|_1$ is given by 
\begin{align}\label{eq:dev_w}
\| W_{[i,:]} \!-\! e_{j_p} \|_1&=(1-W_{i j_p})+\sum\nolimits_{\ell\neq j_p} W_{i \ell} \nonumber \\
& =2(1-W_{i j_p})\triangleq \tau_1^*,
\end{align}
which is exactly the cost for $\|K_{[i,:]}\|_1$ to support a single edge in $\tilde{W}_{[i,:]}$. 
Hence, if the budget $\tau\ge \tau_1^*$, the optimal support number for $\tilde{W}_{[i,]}$ is $c^*(\tau)=1$. 
The first case in \eqref{eq:c1_cost} is proved.

When the available budget $\tau<\tau_1^*$, there will be more non-zero edges in $\tilde{W}_{[i,:]}$ apart from $W_{i j_p}$. 
Recall the non-mandatory candidate support set is given by $\mathcal{C}_i^m = \mathcal{N}_i^+ \setminus\{j_p\} $ when $i\neq r^*$. 
Denote the chosen non-mandatory support set as $\mathcal{C} \subseteq \mathcal{C}_i^m$, and the corresponding weight sum in $W_{[i,:]}$ on $\mathcal{C}$ as 
$\alpha_c=\sum_{j\in\mathcal{C}} W_{ij} $. 
Then, for the rest non-support elements, the partial sum $(1-\alpha_c-W_{i j_p})$ should be compensated to zero. 
Meanwhile, to keep the total sum $\sum_{j=1}^n \tilde{w}_{ij}=1$, the compensated part $(1-\alpha_c-W_{i j_p})$  should be exactly added into $\mathcal{C}\cup\{j_p\} $. 
Hence, the $\ell_1$ distance between $\tilde{W}_{[i,:]}$ and $W_{[i,:]}$ in this case is given by
\begin{equation}
\| \tilde{W}_{[i,:]} - W_{[i,:]} \|_1= 2(1 -\alpha_c -W_{i j_p}). 
\end{equation}
Clearly, the larger $\alpha_c$ is, the smaller $\| \tilde{W}_{[i,:]} - W_{[i,:]} \|_1$ is. 
Given a fixed support number $c=|\mathcal{C}|$, it follows from a counting argument that 
the maximum $\alpha_c$ is given by 
\begin{align}\label{eq:support}
\alpha_c^{*}=\max_{c=|\mathcal{C}|} \alpha_c= \sum\nolimits_{r=1}^c W_{ij_{(r)}}.  
\end{align}
Reversely, when the budget $\tau$ is fixed, the minimal support number is given by
\begin{align}
c_{nroot}^*&=1+\min\left\{ c: 2(1 -\alpha_c^* -W_{i j_p} ) \le \tau \right\} \nonumber \\
  &=1+\min\left\{ c: \alpha_c^* \ge 1-\frac{\tau}{2}-W_{i j_p} \right\} ,
\end{align}
which proves the second case in \eqref{eq:c1_cost}.

Next, we turn to prove \eqref{eq:c2_cost} for $i=r^*$, where the derivation procedures resemble those of case $i\neq r$. 
The only difference is that the two edges $\tilde{w}_{ii}$ and $\tilde{w}_{ij_0}$ are required to be positive, and the support set in \eqref{eq:support} is changed into $\mathcal{C}\subseteq \mathcal{N}_i\setminus\{ {j_0} \cup {i}\}$. 
Hence, the best support number in this case can only be 2, and the corresponding budget threshold is given by
\begin{align}
\tau_2^* \triangleq 2(1-W_{ii}-W_{i j_0}). 
\end{align}
Consequently, when $\tau<\tau_2^*$, more edges are needed to be active, and the minimal edge number is given by 
\begin{align}
c_{root}^*&=2+\min\left\{ c: 2(1 -\alpha_c^* -W_{ii }-W_{i j_0} ) \le \tau \right\} \nonumber \\
  &=2+\min\left\{ c: \alpha_c^* \ge 1-\frac{\tau}{2}-W_{ii}-W_{i j_0} \right\} .
\end{align}
The proof is completed. 
\end{proof}

Theorem \ref{th:solution} gives the optimal values and the corresponding solutions to the problems \eqref{eq:nroot} and \eqref{eq:root}. 
Note that the two optimization problems only seek to find the minimum support number of $\tilde{W}_{[i,:]}$, and has no direct concern on which elements are positive or how large they are. 
The provided solutions are sufficient to reach the optimal values. 
In many situations, the solutions to the problems \eqref{eq:nroot} and \eqref{eq:root} are not unique, due to i) the support set $\mathcal{C}$ could be multiple, 
and ii) how to distribute the part sum $(1-\alpha_c^*)$ is not unique even if $\mathcal{C}$ is unique.


In fact, the proof of Theorem \ref{th:solution} essentially provides a greedy idea to construct closed-form optimal solutions of the problem \eqref{eq:nroot} and \eqref{eq:root}. 
The key insight is that when node $i$ selects a support set of nodes that reach $i$, it just keeps enough of the largest topology weights such that their sum covers $( 1-{\tau}/{2})$. 
Following this line, we design a heuristic topology modification algorithm to solve this step, presented in Algorithm \ref{alg:row_design}. 
Notice that in the inputs, an adjusting parameter $\delta>0$ is used to adapt the non-negativity constraint. 
Since we mainly need to sort the local weights when solving the problem, the overall computational complexity of this heuristic algorithm is $\bm{O}(|\mathcal{N}_i^+| \cdot \log |\mathcal{N}_i^+| )$, which is significantly lower than general MILP problems (NP-hard). 

\begin{remark}
In summary, the proposed topology preservation method starts from a given strongly connected row-stochastic topology, and induces a directed forest structure in a distributed way. 
The topology privacy is preserved from two aspects: i) for edges that can be hidden as permitted by the budget, they are rendered indistinguishable from nonexistent ones, ii) for the remaining edges that are not fully hidden, substantial deviations are introduced to their true weights. 
\end{remark}

\begin{figure}[t]
\centering
\subfigure[ $\tau\!=\!0$ (original topology). ]{\label{topo1}
 \includegraphics[width=0.22\textwidth]{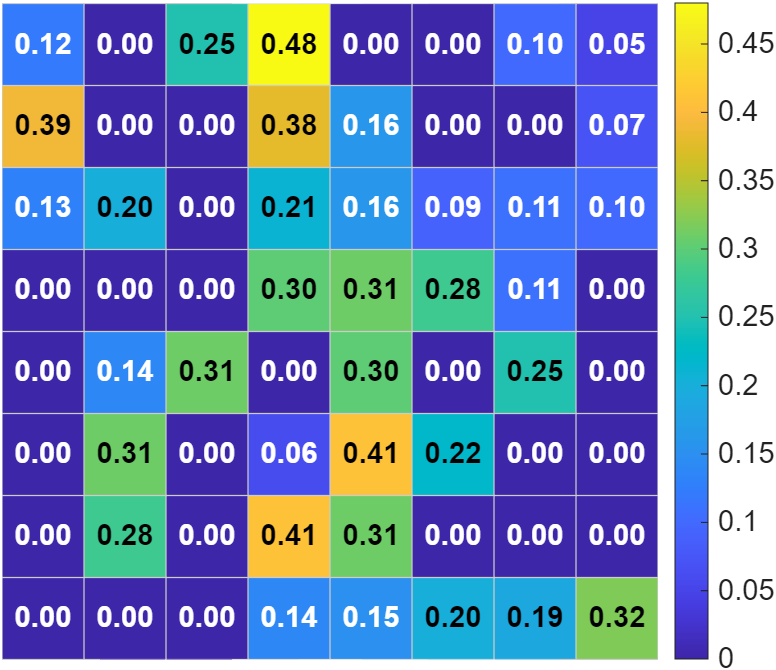}}
\hspace{-0.5pt}
\subfigure[ $\tau=0.4$. ]{\label{topo2}
\includegraphics[width=0.22\textwidth]{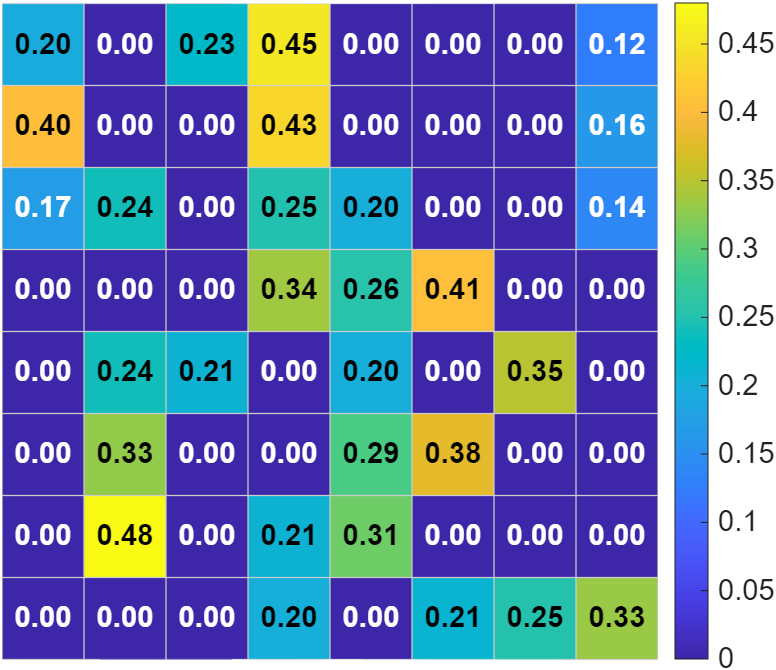}}
\subfigure[ $\tau=0.8$. ]{\label{topo3}
\includegraphics[width=0.22\textwidth]{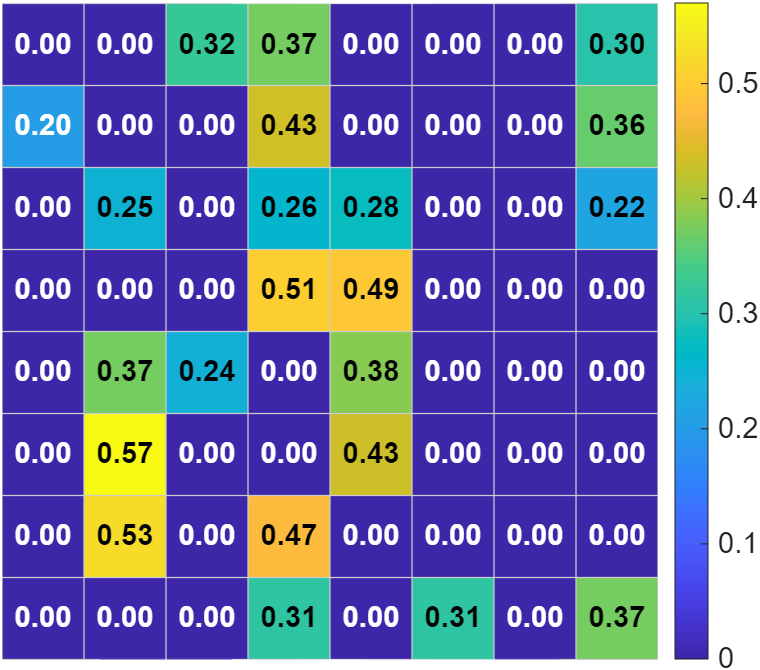}}
\hspace{-0.5pt}
\subfigure[ $\tau=1.2$. ]{\label{top4}
\includegraphics[width=0.22\textwidth]{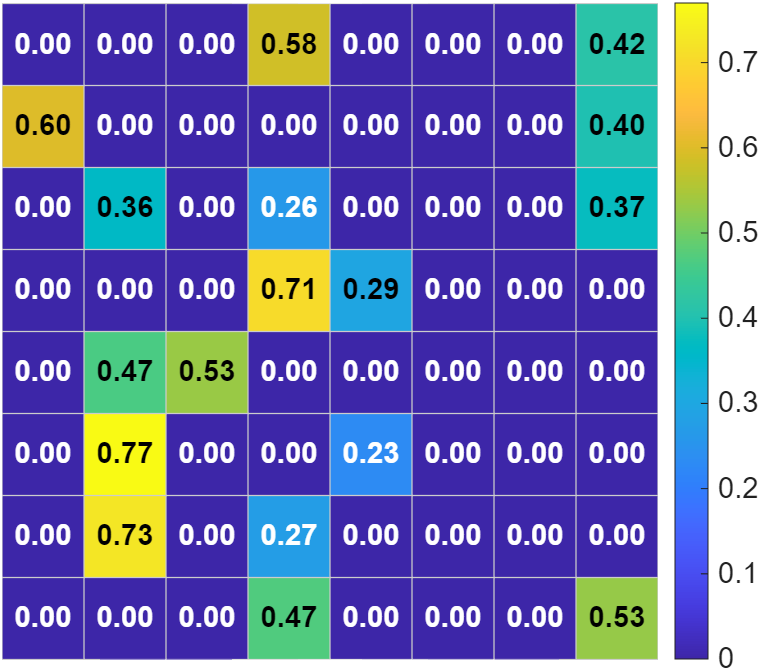}}
\subfigure[ $\tau=1.6$. ]{\label{top5}
\includegraphics[width=0.22\textwidth]{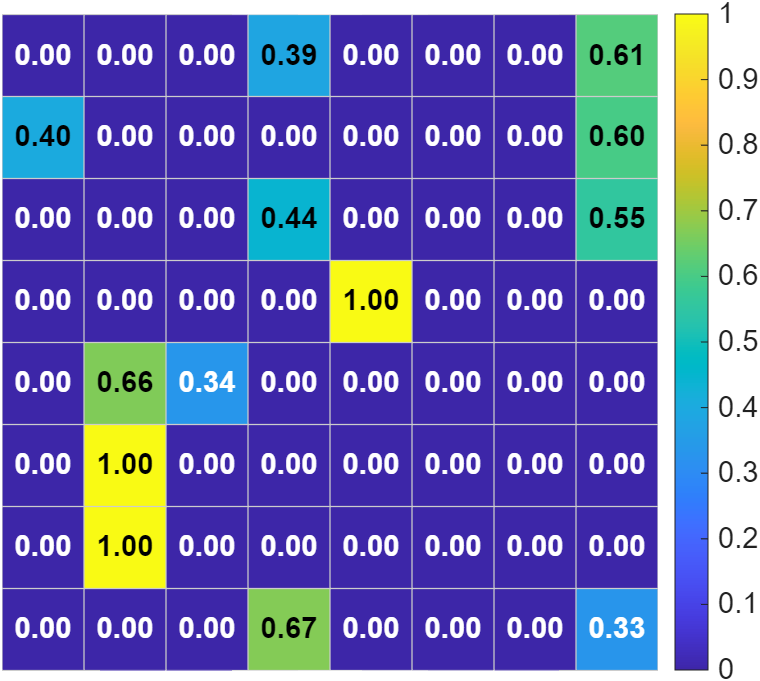}}
\hspace{-0.5pt}
\subfigure[ $\tau=2$. ]{\label{topo6}
\includegraphics[width=0.22\textwidth]{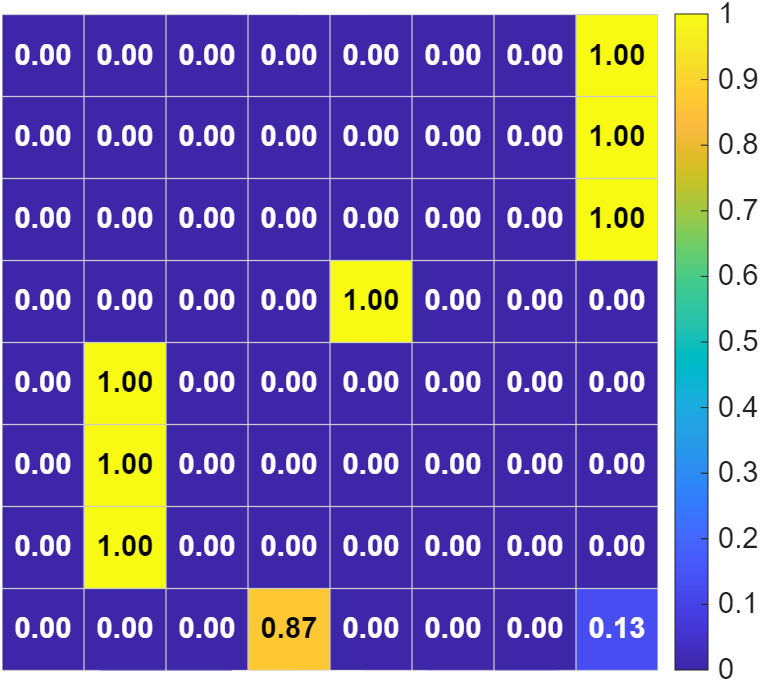}}
\vspace*{-8pt}
\caption{The original topology and the modified topologies given different budgets. Notice that the matrix entries are displayed with two decimal places.}
\label{matrix_W2}
\vspace*{-10pt}
\end{figure}

\begin{figure*}[t]
\centering
\subfigure[Inference performance v.s. $\tau$.]{\label{fig:curve_tau}
\includegraphics[width=0.325\textwidth]{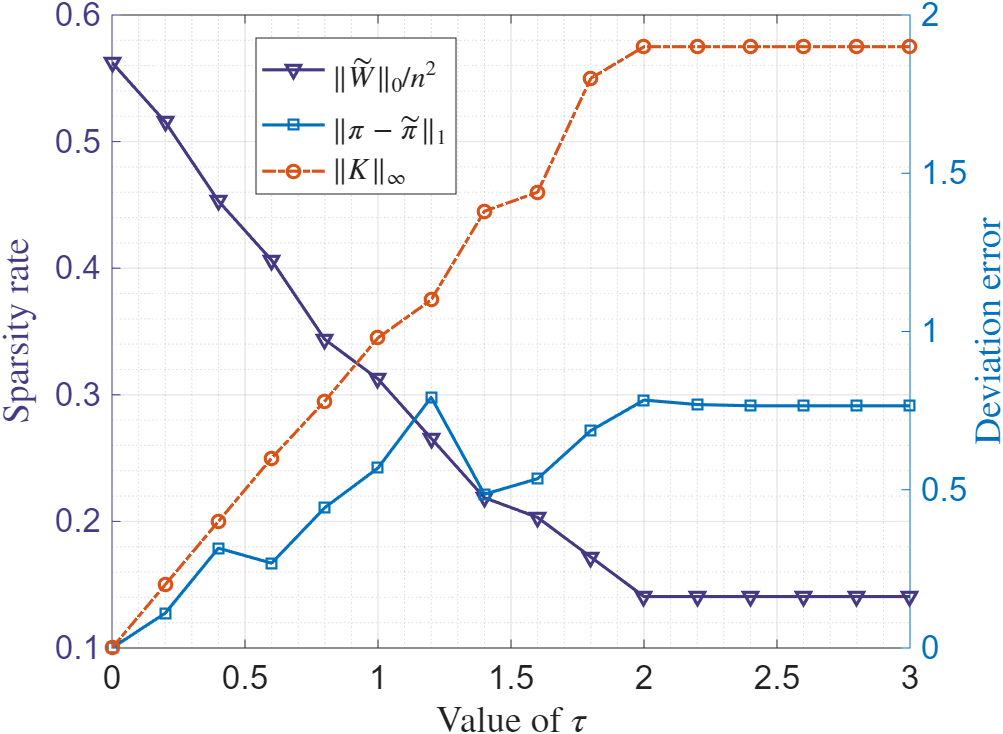}} 
\subfigure[State deviation.]{\label{fig:three_state_error1}
\includegraphics[width=0.315\textwidth]{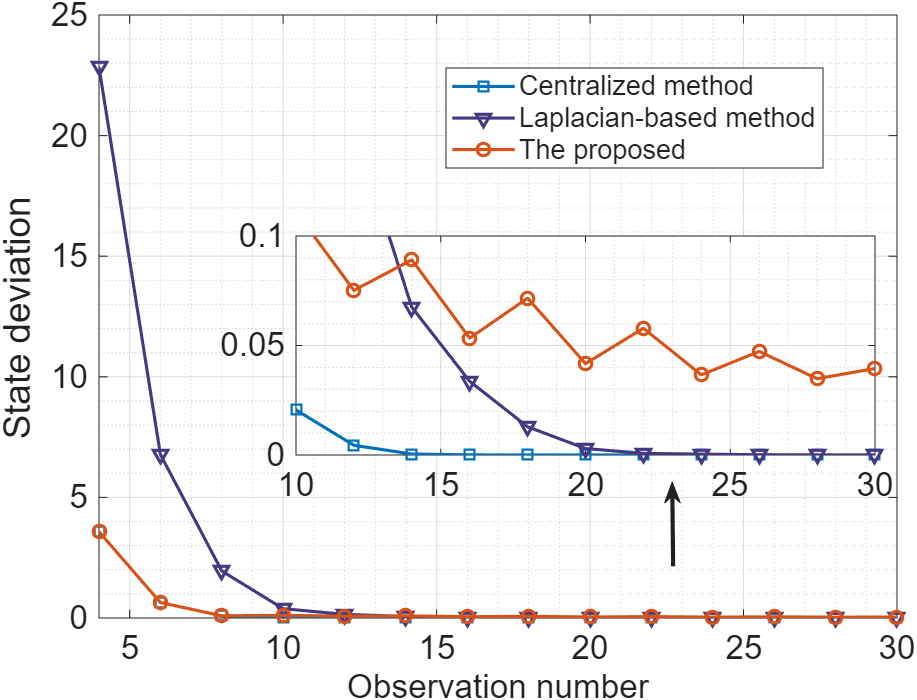}}
\subfigure[Inference errors $Er_1$ and $Er_2$.]{\label{fig:three_topo_error1}
\includegraphics[width=0.31\textwidth]{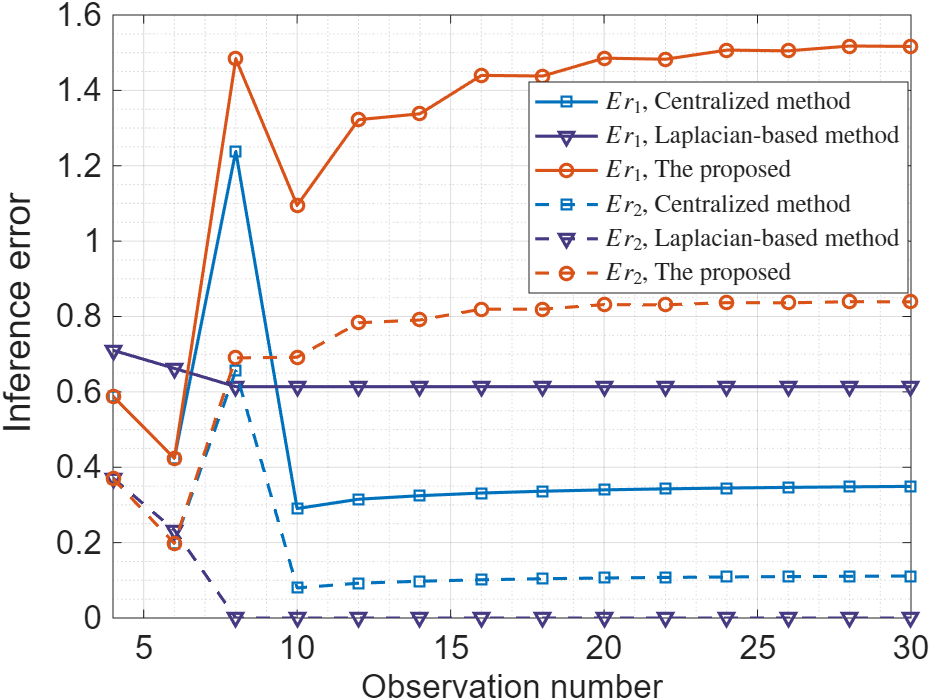}}
\vspace{-8pt}
\caption{Inference performance under different $\tau$ and comparisons with centralized and Laplacian-based methods.}
\label{fig:three_comparison}
\end{figure*}

\section{Numerical Simulations}\label{sec:simulation}

This section provides numerical experiments to verify the effectiveness of the proposed method on rendering topology inaccuracy, and compares it with existing methods. 
For simple expressions, denote the topology identifiable from $\{x_t\}_{t=0}^T$ as $\hat{W}=X_{1:T} (X_{0:T-1})^{\dag}$. 
We consider a directed network of $8$ nodes, whose topology matrix $W$ is drawn in Fig.~\ref{topo1}. 
The initial state is given by $x_0=[-2, -48, -35, -50, -56, 60, 0, -84]^\intercal$, and node $8$ has the maximum beacon value by the randomized initialization.

\textit{Performance under different budget $\tau$}. 
First, we implement the proposed heuristic topology modification algorithm (Algorithm \ref{alg:row_design}) to show how the topology structure changes according to different budget $\tau$, which is presented in Fig.~\ref{topo2}-Fig.~\ref{topo6}. 
It is clear that as $\tau$ increases, a more sparse structure will be promoted while the network connectivity is preserved. 
If the budget $\tau$ is at most sufficient to hide $m$ edges (determined by the critical conditions in Theorem \ref{th:solution}) and still has remaining parts, then the algorithm will distribute remaining budgets to other nonzero elements and increase their deviations with the ground truth. 
A typical example is the case of $W_{[1,:]}$ when $\tau=0.4$. 
In this case, the parent node of node $1$ is selected as $j_p=8$ according to the criteria $j_p\in \arg\min_{ j\in\mathcal{N}_i } d_j^{(n-1)} $, and thus $W_{18}$ cannot be made zero. 
Then, the budget only supports to hide the edge $W_{17}=0.1$ along with its compensation, and the remaining budget $(\tau-2 W_{17})$ are distributed to other edges. 
When $\tau=2$, this budget is sufficiently large to reach to the optimal values in Problem \eqref{eq:root} and \eqref{eq:nroot}, where the new root node $8$ maintains two edges while the rest nodes maintain only one, as shown in \ref{topo6}. 
Besides, we plot the sparsity rate of the topology and the bounds $\{\|K\|_{\infty}, \|\pi -\tilde{\pi}\|_{\infty}\}$ as $\tau$ increases. 
As shown in Fig.~\ref{fig:curve_tau}, the sparsity rate will decrease and the bound $\|K \|_{\infty}$ will increase to stable values as $\tau$ is sufficiently large. 
The error $ \|\pi -\tilde{\pi}\|_{\infty}$ also becomes stable as $\tau$ grows. 
We remark that the stable phenomenon is expected because when the budget is sufficient large, the modified topology will be uniquely determined by Algorithm \ref{alg:row_design}. 

\textit{Comparisons with centralized and Laplacian-based methods}. 
Next, given the same initial state $x_0$ and budget $\tau=2$, we compare the proposed method with the centralized method demonstrated in Sec.~\ref{subsec:inaccurate} (denoted as \textbf{M1}) and the Laplacian-based method \cite{yushan2025Resistant} (denoted as \textbf{M2}). 
The state deviation curves are plotted in Fig.~\ref{fig:three_state_error1}. 
It shows that as the observation number $T$ increases, the deviation $\|x_t-x_t^*\|_2$ for both \textbf{M1} and \textbf{M2} will approach to zero, 
while it is maintained below a small value for the proposed distributed method. 
This difference is expected because the proposed method slightly sacrifices the accuracy of $\|x_t-x_t^*\|_2$ to promote a sparsity structure. 
Then, we draw the topology inference errors of three methods, which are evaluated by the following normalized mean square error (denoted as $Er_1$) and its tailored version for the non-diagonal elements of $W$ (denoted as $Er_2$):
\begin{align*}
Er_1 = \frac{\| \hat{W}- W\|_{F}}{\|W\|_{F}},~Er_2 = \frac{ \min_{\gamma>0}\| \gamma \hat{W}_{\text{nd}}- W_{\text{nd}}\|_{F}}{\|W_{\text{nd}}\|_{F}},
\end{align*}
\!\!where $(W_{\text{nd}})_{ij}=W_{ij}$ for $i\neq j$ and $(W_{\text{nd}})_{ii}=0$ for $i\in\mathcal{V}$. 
Note that the parameter $\gamma$ in $Er_2$ is used to avoid the scalar ambiguity in $\textbf{M2}$. 
As shown in Fig.~\ref{fig:three_topo_error1}, \textbf{M1} and the proposed method incur biased inference error concerning both $Er_1$ and $Er_2$, while \textbf{M2} can only achieve that concerning $Er_1$. 
This point verifies the strength of the proposed method in preserving relative weight information, as analyzed in Sec.~\ref{subsec:Laplacian}.

\begin{figure}[t]
\centering
\subfigure[State deviation.]{\label{fig:noise_state}
\includegraphics[width=0.232\textwidth]{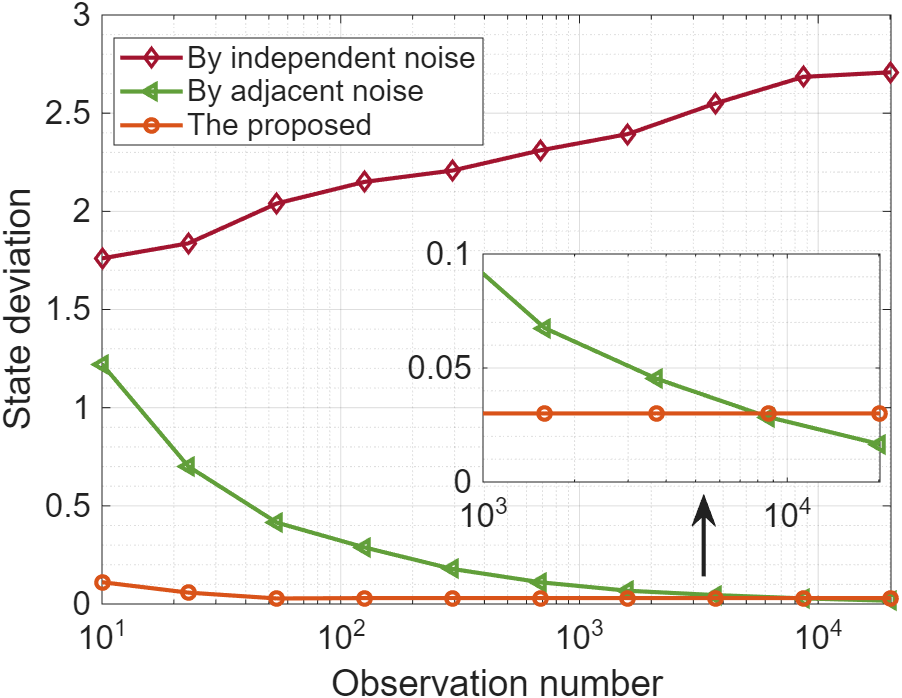}}
\hspace{-8pt}
\subfigure[Inference error $Er_1$.]{\label{fig:noise_topo}
\includegraphics[width=0.232\textwidth]{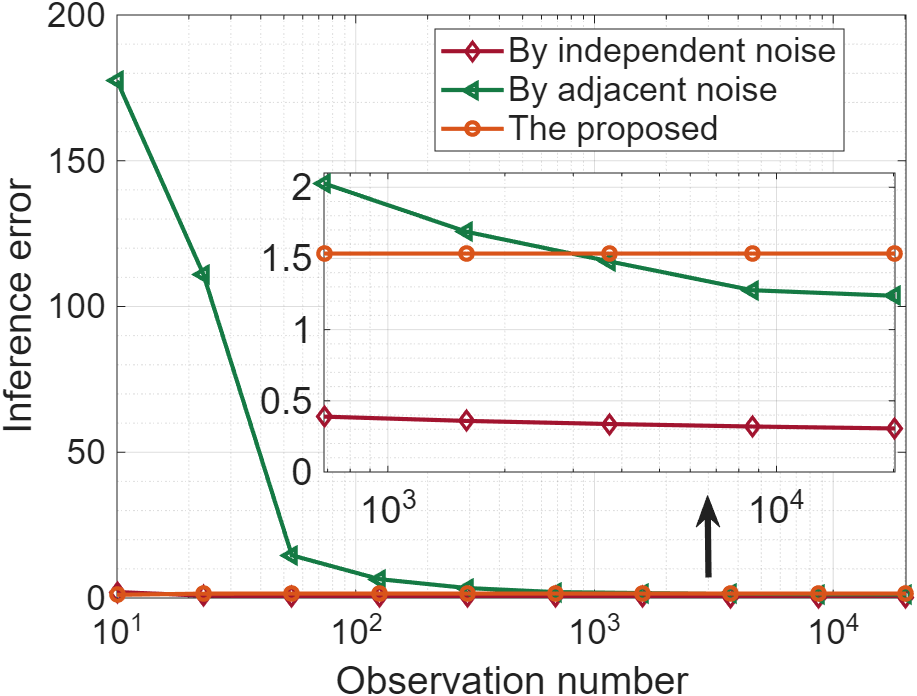}}
\vspace{-8pt}
\caption{Comparisons with noise-adding based methods.}
\label{fig:noise_comparison}
\end{figure}

\textit{Comparisons with noise-adding based methods}. 
Finally, given the same setting as the last experiment, we turn to examine the comparisons between the proposed method with existing noise-adding methods in private consensus literature \cite{mo2017privacy,he2020differential,li2025Preserving}. 
The basic idea of the latter methods is to use adjacent decaying noises in the iteration process (denoted as \textbf{M3}), which can be described by
\begin{align}
    x^a_{t+1}=W x^a_t + \theta_t,~\theta_t=\omega_t - \omega_{t-1},
\end{align}
where $\omega_t $ is independent zero-mean noise with decaying variance and satisfies $\lim_{t\to\infty} \mathbb{E}[\omega_t \omega_t^\intercal]=0$. 
Notice that the reason of making $\theta_t$ contain two adjacent decaying noises is to achieve $\lim_{t\to\infty}\mathbb{E}[\|x^a_t -x^*_t\|_2]=0$, and we use $\mathbb{E}[\omega_t \omega_t^\intercal]=I_n/(t+1)^2$ in this experiment. 
Since we emphasize a worse-case preservation, the inference method here adopts the following asymptotic unbiased estimator \cite[Chap. 7.6]{ljung1999system}:
\begin{align}
    \hat{W}_{adj}=\Gamma(2) \Gamma^{-1}(1),
\end{align}
where $\Gamma(1)=\frac{1}{T-1} \sum_{t=2}^T x^a_t (x^a_{t-1})^\intercal$ and  $\Gamma(2)=\frac{1}{T-2} \sum_{t=3}^T x^a_t (x^a_{t-2})^\intercal$. 
For better illustration, we also consider the case of using independent decaying noise (denoted as \textbf{M4}), i.e., $x^b_{t+1}=W x^b_t+\omega_t$, and the corresponding inference method adopts the ordinary least square estimator, which is also asymptotically unbiased in this case \cite{li2025Preserving}. 
Due to the stochasticity in \textbf{M3} and \textbf{M4}, we randomly run the experiments $50$ times for the two methods and use the average error for evaluation. 
Meanwhile, to exhibit the asymptotic performance, we use log-axis for the observation number. 
As depicted in Fig.~\ref{fig:noise_state}, the state deviation curve for \textbf{M3} will approach zero as $T\to\infty$, while the curve for \textbf{M4} fluctuates as $T$ increases due to $\lim_{t\to\infty}\mathbb{E}[\|x^b_t -x^*_t\|_2]\neq 0$. 
By contrast, the proposed method only incurs a small stable deviation (below $0.05$), which is controlled by $\tau$ and acceptable in practice. 
Then, Fig.~\ref{fig:noise_topo} plots the topology inference error $Er_1$ for three methods. 
It is observed that the curves for \textbf{M3} and \textbf{M4} will gradually decay to zero in a slow rate, while the curve for the proposed remains stable regardless of $T$. 
This point demonstrates the outperformance of the proposed method in preserving the topology.

\section{Conclusions}\label{sec:conclusion}

In this paper, we investigated the topology privacy preservation problem in consensus networks. 
Specifically, we proposed a feedback-based design that intentionally violates topology identifiability conditions, thereby preventing external observers from recovering the true topology from available data, while preserving the intended consensus convergence. 
Considering both the partial and full observation cases, we derived the conditions of feedback to enforce topology unsolvability or inaccuracy from available observations. 
The derived solution constructions can be directly applied a class of consensus networks with centralized–design and distributed–execution architecture. 
Furthermore, for the distributed design setting, we established a controllable tradeoff between the state convergence deviation and the preservation of existing links, and proposed a topology modification method with limited privacy budget. 
Compared with noise-adding based methods, this method enjoyed the merit of finite-time implementations and controllable privacy guarantees. 
Extensive simulations demonstrated the effectiveness of the proposed method. 

Our future work will focus on developing an adaptive distributed scheme, which supports individual nodes to learn global topology information while achieving stronger topology privacy preservation performance.

\appendix

\subsection{Topology Identification from Partial Observations}\label{appen:inference}

In this part, we show how to infer the topology indirectly using partial observation data. 
Similar algorithms have been developed for the realization problem \cite{Ho1966effective,He2025range}. 
Recall the constructed matrices $\tilde{Y}_{0:T-n}^{*}$ given in \eqref{eq:Hankel} and $X_{0:T-n}^{*}$ given in \eqref{eq:old_estimator}, 
and it holds that $\tilde{Y}_{0:T-n}^{*} = Q_o^* X_{0:T-n}^{*}$.  
Let two submatrices of $\tilde{Y}_{0:T-n}^{*}$ be
\begin{equation*}
\begin{split}
    \tilde{Y}_{0:T-n}^{*,-} &=[\tilde{y}_{0:n-1}^{*},\tilde{y}_{1:n}^{*},\cdots,\tilde{y}_{T-n-1:T-2}^{*}] \in \mathbb{R}^{n m \times  (T-n)}, \\
    \tilde{Y}_{0:T-n}^{*,+} &=[\tilde{y}_{1:n}^{*},\tilde{y}_{2:n+1}^{*},\cdots,\tilde{y}_{T-n:T-1}^{*}] \in \mathbb{R}^{n m \times  (T-n)},
\end{split}    
\end{equation*}
and they can be rewritten as
\begin{align}
  \tilde{Y}_{0:T-n}^{*,-} = Q_o^* X_{0:T-n-1}^{*}, ~\tilde{Y}_{0:T-n}^{*,+} = Q_o^* WX_{0:T-n-1}^{*}.  \label{eq:yt2}
\end{align}
Taking the singular value decomposition on $\tilde{Y}_{0:T-n}^{*,-}$, we have
\begin{equation} \label{eq:KF_past2}
  \tilde{Y}_{0:T-n}^{*,-} = U\Lambda V^\intercal = U_n\Lambda_n V_n^\intercal,
\end{equation}
where $\Lambda_n = \text{diag}\left(\sigma_1,\sigma_2,\dots,\sigma_{n}\right)$ is a diagonal matrix, having the first $n$ non-zero singular values of $\Lambda$ on its diagonal and satisfying $\sigma_1 \geq \sigma_2 \geq \cdots \geq \sigma_{n}$, 
$U_n$ and $V_n$ contain the $n$ groups of left and right singular vectors corresponding to $\Lambda_n$, respectively. 
The last equality in \eqref{eq:KF_past2} is due to that $\tilde{Y}_{0:T-n}^{*,-}$ has rank equal to $n$. 
Based on the product structure of $\tilde{Y}_{0:T-n}^{*,-}$ in \eqref{eq:yt2}, $Q_o^*$ and $X_{0:T-n-1}^{*}$, up to a similarity transformation, are given by
\begin{equation} \label{eq:svd}
    \hat Q_o^* = U_n\Lambda_n^{\frac{1}{2}},  \quad   \hat X_{0:T-n-1}^{*}= \Lambda_n^{\frac{1}{2}} V_n^\intercal.  
\end{equation}
Finally, based on \eqref{eq:yt2}, an estimate of $W$ up to a similarity transformation is obtained by
\begin{equation}
    Q_W = (\hat Q_o^*)^\dagger \tilde{Y}_{0:T-n}^{*,+} (\hat X_{0:T-n-1}^{*})^\dagger.
\end{equation}

\vspace*{-15pt}

\bibliographystyle{IEEEtran}

\begin{thebibliography}{10}
\providecommand{\url}[1]{#1}
\csname url@samestyle\endcsname
\providecommand{\newblock}{\relax}
\providecommand{\bibinfo}[2]{#2}
\providecommand{\BIBentrySTDinterwordspacing}{\spaceskip=0pt\relax}
\providecommand{\BIBentryALTinterwordstretchfactor}{4}
\providecommand{\BIBentryALTinterwordspacing}{\spaceskip=\fontdimen2\font plus
\BIBentryALTinterwordstretchfactor\fontdimen3\font minus
  \fontdimen4\font\relax}
\providecommand{\BIBforeignlanguage}[2]{{%
\expandafter\ifx\csname l@#1\endcsname\relax
\typeout{** WARNING: IEEEtran.bst: No hyphenation pattern has been}%
\typeout{** loaded for the language `#1'. Using the pattern for}%
\typeout{** the default language instead.}%
\else
\language=\csname l@#1\endcsname
\fi
#2}}
\providecommand{\BIBdecl}{\relax}
\BIBdecl

\bibitem{olfati2007consensus}
R.~Olfati-Saber, J.~A. Fax, and R.~M. Murray, ``Consensus and cooperation in
  networked multi-agent systems,'' \emph{Proceedings of the IEEE}, vol.~95,
  no.~1, pp. 215--233, 2007.

\bibitem{brugere2018network}
I.~Brugere, B.~Gallagher, and T.~Y. Berger-Wolf, ``Network structure inference,
  a survey: Motivations, methods, and applications,'' \emph{ACM Computing
  Surveys (CSUR)}, vol.~51, no.~2, pp. 1--39, 2018.

\bibitem{shvydun2024system}
S.~Shvydun and P.~Van~Mieghem, ``System identification for temporal networks,''
  \emph{IEEE Transactions on Network Science and Engineering}, vol.~11, no.~2,
  pp. 1885--1895, 2024.

\bibitem{hou2021combating}
T.~Hou, T.~Wang, Z.~Lu, and Y.~Liu, ``Combating adversarial network topology
  inference by proactive topology obfuscation,'' \emph{IEEE/ACM Transactions on
  Networking}, vol.~29, no.~6, pp. 2779--2792, 2021.

\bibitem{10210275}
B.~Shan, X.~Yuan, W.~Ni, X.~Wang, R.~P. Liu, and E.~Dutkiewicz, ``Preserving
  the privacy of latent information for graph-structured data,'' \emph{IEEE
  Transactions on Information Forensics and Security}, vol.~18, pp. 5041--5055,
  2023.

\bibitem{9249412}
E.~Testi and A.~Giorgetti, ``Blind wireless network topology inference,''
  \emph{IEEE Transactions on Communications}, vol.~69, no.~2, pp. 1109--1120,
  2021.

\bibitem{sebastian2023learning}
E.~Sebasti{\'a}n, T.~Duong, N.~Atanasov, E.~Montijano, and C.~Sag{\"u}{\'e}s,
  ``Learning to identify graphs from node trajectories in multi-robot
  networks,'' in \emph{2023 International Symposium on Multi-Robot and
  Multi-Agent Systems (MRS)}.\hskip 1em plus 0.5em minus 0.4em\relax IEEE,
  2023, pp. 142--148.

\bibitem{mateos2019connecting}
G.~Mateos, S.~Segarra, A.~G. Marques, and A.~Ribeiro, ``Connecting the dots:
  Identifying network structure via graph signal processing,'' \emph{IEEE
  Signal Processing Magazine}, vol.~36, no.~3, pp. 16--43, 2019.

\bibitem{10146241}
G.~Leus, A.~G. Marques, J.~M. Moura, A.~Ortega, and D.~I. Shuman, ``Graph
  signal processing: History, development, impact, and outlook,'' \emph{IEEE
  Signal Processing Magazine}, vol.~40, no.~4, pp. 49--60, 2023.

\bibitem{9695344}
Z.~Liu, W.~Wang, G.~Ding, Q.~Wu, and X.~Wang, ``Topology sensing of
  non-collaborative wireless networks with conditional {Granger} causality,''
  \emph{IEEE Transactions on Network Science and Engineering}, vol.~9, no.~3,
  pp. 1501--1515, 2022.

\bibitem{li2024topology}
Y.~Li, J.~He, C.~Chen, and X.~Guan, ``Topology inference for network systems:
  Causality perspective and non-asymptotic performance,'' \emph{IEEE
  Transactions on Automatic Control}, vol.~69, no.~6, pp. 3483--3498, 2024.

\bibitem{ioannidis2019semi}
V.~N. Ioannidis, Y.~Shen, and G.~B. Giannakis, ``Semi-blind inference of
  topologies and dynamical processes over dynamic graphs,'' \emph{IEEE
  Transactions on Signal Processing}, vol.~67, no.~9, pp. 2263--2274, 2019.

\bibitem{zaman2021online}
B.~Zaman, L.~M.~L. Ramos, D.~Romero, and B.~{Beferull-Lozano}, ``Online
  topology identification from vector autoregressive time series,'' \emph{IEEE
  Transactions on Signal Processing}, vol.~69, pp. 210--225, 2021.

\bibitem{hendrickx2019identifiability}
J.~M. Hendrickx, M.~Gevers, and A.~S. Bazanella, ``Identifiability of dynamical
  networks with partial node measurements,'' \emph{IEEE Transactions on
  Automatic Control}, vol.~64, no.~6, pp. 2240--2253, 2019.

\bibitem{cheng2022allocation}
X.~Cheng, S.~Shi, and P.~M.~J. {Van den Hof}, ``Allocation of excitation
  signals for generic identifiability of linear dynamic networks,'' \emph{IEEE
  Transactions on Automatic Control}, vol.~67, no.~2, pp. 692--705, 2022.

\bibitem{sun2024identifiability}
W.~Sun, J.~Xu, and J.~Chen, ``Identifiability and identification of switching
  dynamical networks: A data-based approach,'' \emph{IEEE Transactions on
  Control of Network Systems}, vol.~11, no.~3, pp. 1177--1189, 2024.

\bibitem{nabi-abdolyousefi2012networka}
M.~{Nabi-Abdolyousefi} and M.~Mesbahi, ``Network identification via node
  knockout,'' \emph{IEEE Transactions on Automatic Control}, vol.~57, no.~12,
  pp. 3214--3219, 2012.

\bibitem{10337619}
E.~Restrepo, N.~Wang, and D.~V. Dimarogonas, ``Simultaneous topology
  identification and synchronization of directed dynamical networks,''
  \emph{IEEE Transactions on Control of Network Systems}, vol.~11, no.~3, pp.
  1491--1501, 2024.

\bibitem{vanwaarde2021topology}
H.~J. {van Waarde}, P.~Tesi, and M.~K. Camlibel, ``Topology identification of
  heterogeneous networks: Identifiability and reconstruction,''
  \emph{Automatica}, vol. 123, p. 109331, 2021.

\bibitem{mo2017privacy}
Y.~Mo and R.~M. Murray, ``Privacy preserving average consensus,'' \emph{IEEE
  Transactions on Automatic Control}, vol.~62, no.~2, pp. 753--765, 2017.

\bibitem{nozari2017differentially}
E.~Nozari, P.~Tallapragada, and J.~Cort{\'e}s, ``Differentially private average
  consensus: Obstructions, trade-offs, and optimal algorithm design,''
  \emph{Automatica}, vol.~81, pp. 221--231, 2017.

\bibitem{he2020differential}
J.~He, L.~Cai, and X.~Guan, ``Differential private noise adding mechanism and
  its application on consensus algorithm,'' \emph{IEEE Transactions on Signal
  Processing}, vol.~68, pp. 4069--4082, 2020.

\bibitem{10261255}
E.~Rizk, S.~Vlaski, and A.~H. Sayed, ``Enforcing privacy in distributed
  learning with performance guarantees,'' \emph{IEEE Transactions on Signal
  Processing}, vol.~71, pp. 3385--3398, 2023.

\bibitem{wang2021dynamic}
X.~Wang, H.~Ishii, J.~He, and P.~Cheng, ``Dynamic privacy-aware collaborative
  schemes for average computation: A multi-time reporting case,'' \emph{IEEE
  Transactions on Information Forensics and Security}, vol.~16, pp. 3843--3858,
  2021.

\bibitem{9999003}
W.~Zhang, Z.~Zuo, Y.~Wang, and G.~Hu, ``How much noise suffices for privacy of
  multiagent systems?'' \emph{IEEE Transactions on Automatic Control}, vol.~68,
  no.~10, pp. 6051--6066, 2023.

\bibitem{10360869}
Q.~Li, J.~S. Gundersen, M.~Lopuhaä-Zwakenberg, and R.~Heusdens, ``Adaptive
  differentially quantized subspace perturbation {(ADQSP)}: A unified framework
  for privacy-preserving distributed average consensus,'' \emph{IEEE
  Transactions on Information Forensics and Security}, vol.~19, pp. 1780--1793,
  2024.

\bibitem{katewa2020differential}
V.~Katewa, A.~Chakrabortty, and V.~Gupta, ``Differential privacy for network
  identification,'' \emph{IEEE Transactions on Control of Network Systems},
  vol.~7, no.~1, pp. 266--277, 2020.

\bibitem{Hawkins2024node}
C.~Hawkins, B.~Chen, K.~Yazdani, and M.~Hale, ``Node and edge differential
  privacy for graph {Laplacian} spectra: Mechanisms and scaling laws,''
  \emph{IEEE Transactions on Network Science and Engineering}, vol.~11, no.~2,
  pp. 1690--1701, 2024.

\bibitem{li2025Preserving}
Y.~Li, Z.~Wang, J.~He, C.~Chen, and X.~Guan, ``Preserving topology of network
  systems: Metric, analysis, and optimal design,'' \emph{IEEE Transactions on
  Automatic Control}, vol.~70, no.~6, pp. 3540--3555, 2025.

\bibitem{mao2025unidentifiability}
X.~Mao and J.~He, ``Unidentifiability of system dynamics: Conditions and
  controller design,'' \emph{IEEE Transactions on Automatic Control}, vol.~70,
  no.~2, pp. 1380--1387, 2025.

\bibitem{hao2024discernibility}
Y.~Hao, Q.~Wang, Z.~Duan, and G.~Chen, ``Discernibility of topological
  variations for networked {LTI} systems based on observed output
  trajectories,'' \emph{Automatica}, vol. 163, p. 111547, 2024.

\bibitem{fan2024output}
Z.~Fan, X.~Wu, B.~Mao, and J.~L{\"u}, ``Output discernibility of topological
  variations in linear dynamical networks,'' \emph{IEEE Transactions on
  Automatic Control}, vol.~69, no.~8, pp. 5524--5530, 2024.

\bibitem{yushan2025Resistant}
Y.~Li, J.~He, and D.~V. Dimarogonas, ``Resistant topology inference in
  consensus networks: A feedback-based design,'' in \emph{2025 IEEE 64th
  Conference on Decision and Control (CDC)}.\hskip 1em plus 0.5em minus
  0.4em\relax IEEE, 2025, pp. 2563--2568.

\bibitem{FB-LNS}
F.~Bullo, \emph{Lectures on Network Systems}.\hskip 1em plus 0.5em minus
  0.4em\relax Kindle Direct Publishing, Edition 1.6, 2022.

\bibitem{ljung1999system}
L.~Ljung, ``System identification: Theory for the user.''\hskip 1em plus 0.5em
  minus 0.4em\relax Prentice Hall PTR, Second Edition, 1999.

\bibitem{van2023informativity}
H.~J. Van~Waarde, J.~Eising, M.~K. Camlibel, and H.~L. Trentelman, ``The
  informativity approach: To data-driven analysis and control,'' \emph{IEEE
  Control Systems Magazine}, vol.~43, no.~6, pp. 32--66, 2023.

\bibitem{abawajy2016privacy}
J.~H. Abawajy, M.~I.~H. Ninggal, and T.~Herawan, ``Privacy preserving social
  network data publication,'' \emph{IEEE Communications Surveys \& Tutorials},
  vol.~18, no.~3, pp. 1974--1997, 2016.

\bibitem{Wang2024Topology}
Z.~Wang, Y.~Li, X.~Duan, and J.~He, ``Topology-preserving motion coordination
  for multi-robot systems in adversarial environments,'' \emph{IEEE Journal of
  Selected Topics in Signal Processing}, vol.~18, no.~3, pp. 473--486, 2024.

\bibitem{kato2013perturbation}
T.~Kato, \emph{Perturbation theory for linear operators}.\hskip 1em plus 0.5em
  minus 0.4em\relax Springer Science \& Business Media, 2013, vol. 132.

\bibitem{Gohberg2006invariant}
I.~Gohberg, P.~Lancaster, and L.~Rodman, \emph{Invariant subspaces of matrices
  with applications}.\hskip 1em plus 0.5em minus 0.4em\relax SIAM, 2006.

\bibitem{horn2012matrix}
R.~A. Horn and C.~R. Johnson, \emph{Matrix analysis}.\hskip 1em plus 0.5em
  minus 0.4em\relax Cambridge university press, 2012.

\bibitem{7004042}
Y.~Gu, F.~Ren, Y.~Ji, and J.~Li, ``The evolution of sink mobility management in
  wireless sensor networks: A survey,'' \emph{IEEE Communications Surveys \&
  Tutorials}, vol.~18, no.~1, pp. 507--524, 2016.

\bibitem{8892668}
Y.~Khayat, Q.~Shafiee, R.~Heydari, M.~Naderi, T.~Dragičević, J.~W.
  Simpson-Porco, F.~Dörfler, M.~Fathi, F.~Blaabjerg, J.~M. Guerrero, and
  H.~Bevrani, ``On the secondary control architectures of ac microgrids: An
  overview,'' \emph{IEEE Transactions on Power Electronics}, vol.~35, no.~6,
  pp. 6482--6500, 2020.

\bibitem{alonso2019distributed}
J.~Alonso-Mora, E.~Montijano, T.~N{\"a}geli, O.~Hilliges, M.~Schwager, and
  D.~Rus, ``Distributed multi-robot formation control in dynamic
  environments,'' \emph{Autonomous Robots}, vol.~43, no.~5, pp. 1079--1100,
  2019.

\bibitem{cho2001comparison}
G.~E. Cho and C.~D. Meyer, ``Comparison of perturbation bounds for the
  stationary distribution of a {Markov} chain,'' \emph{Linear Algebra and Its
  Applications}, vol. 335, no. 1-3, pp. 137--150, 2001.

\bibitem{hirvonen2020distributed}
J.~Hirvonen and J.~Suomela, ``Distributed algorithms 2020,'' \emph{Aalto
  University, Finland: online}, 2020.

\bibitem{Ho1966effective}
B.~HO and R.~E. K{\'a}lm{\'a}n, ``Effective construction of linear
  state-variable models from input/output functions,''
  \emph{at-Automatisierungstechnik}, vol.~14, no. 1-12, pp. 545--548, 1966.

\bibitem{He2025range}
J.~He, Y.~Xu, Y.~Ju, C.~R. Rojas, and H.~Hjalmarsson, ``Range space or null
  space: Least-squares methods for the realization problem,'' \emph{arXiv
  preprint arXiv:2505.19639}, 2025.

\end{thebibliography}

\end{document}